\documentclass[11pt,fleqn]{article}
\usepackage{graphicx}
\usepackage{setspace}
\usepackage{amsmath}
\usepackage{amsthm}
\usepackage{amsfonts}
\usepackage{multirow}
\usepackage{booktabs}
\usepackage{amssymb}
\usepackage{authblk}
\usepackage{url}
\usepackage{bm}

\RequirePackage[comma]{natbib}
\newcommand{\av}{\mbox{\boldmath$\alpha$}}

\newcommand{\bv}{\mbox{\boldmath$\beta$}}
\newcommand{\sbv}{\mbox{\boldmath$b$}}
\newcommand{\bvs}{\mbox{\boldmath$\beta$}_{\rm sys}}
\newcommand{\bva}{\mbox{\boldmath$\beta$}_{\rm all}}
\newcommand{\hbva}{\mbox{\boldmath$\hat \beta$}_{\rm all}}
\newcommand{\hbvs}{\mbox{\boldmath$\hat \beta$}_{\rm sys}}
\newcommand{\bvc}{\mbox{\boldmath$\beta$}_c}
\newcommand{\bvg}{\mbox{\boldmath$\beta$}_g}

\newcommand{\tbvc}{\mbox{\boldmath$\tilde \beta$}_c}
\newcommand{\hbvg}{\mbox{\boldmath$\hat \beta$}_g}

\newcommand{\ev}{\mbox{\boldmath$e$}}

\newcommand{\Bv}{\mbox{\boldmath$B$}}

\newcommand{\Dv}{\mbox{\boldmath$D$}}
\newcommand{\Ev}{\mbox{\boldmath$E$}}

\newcommand{\Gv}{\mbox{\boldmath$G$}}

\newcommand{\Iv}{\mbox{\boldmath$I$}}

\newcommand{\gav}{\mbox{\boldmath$\gamma$}}
\newcommand{\Pv}{\mbox{\boldmath$P$}}
\newcommand{\Qv}{\mbox{\boldmath$Q$}}

\newcommand{\SSv}{\mbox{\boldmath$S$}}
\newcommand{\Sigv}{\mbox{\boldmath$\Sigma$}}
\newcommand{\SSigv}{\mbox{\boldmath$\mathcal{E}$}}
\newcommand{\hSigv}{\mbox{\boldmath$\check{\mathcal{E}}$}}

\newcommand{\Xv}{\mbox{\boldmath$X$}}
\newcommand{\XXv}{\mbox{\boldmath$\mathcal{X}$}}

\newcommand{\Uv}{\mbox{\boldmath$U$}}
\newcommand{\Wv}{\mbox{\boldmath$W$}}

\newcommand{\Vv}{\mbox{\boldmath$V$}}

\newcommand{\Yv}{\mbox{\boldmath$Y$}}
\newcommand{\YYv}{\mbox{\boldmath$\mathcal{Y}$}}
\newcommand{\Zv}{\mbox{\boldmath$Z$}}

\newcommand{\lv}{\mathbb{I}}

\newcommand{\pxi}{\ensuremath{\Pr\left(\xi(\bvg)\right)}}
\newcommand{\Psiv}{\mbox{\boldmath$\Psi$}}
\newcommand{\Hv}{\mbox{\boldmath$H$}}
\newcommand{\Gammav}{\mbox{\boldmath$\Gamma$}}

\newcommand{\Phiv}{\mbox{\boldmath$\Phi$}}
\newcommand{\BBv}{\mbox{\boldmath$\mathcal{B}$}}
\newcommand{\Omegav}{\mbox{\boldmath$\Omega$}}
\newcommand{\Av}{\mbox{\boldmath$A$}}

\newcommand{\BF}{\ensuremath{{\rm BF}}}
\newcommand{\ABF}{\ensuremath{{\rm ABF}}}

\newcommand{\vect}{\ensuremath{{\rm vec}}}
\newcommand{\etr}{\ensuremath{{\rm etr}}}

\newtheorem{prop}{PROPOSITION}
\newtheorem{lemma}{LEMMA}
\newtheorem{note}{NOTE}
\newtheorem{defn}{DEFINITION}

\title{{\bf Bayesian Model Selection in Complex Linear Systems, as Illustrated in Genetic Association Studies}}

\author
{Xiaoquan Wen \\
Department of Biostatistics,  University of Michigan, Ann Arbor, USA}

\begin{document}

\maketitle

\begin{abstract}
Motivated by examples from genetic association studies, this paper considers the model selection problem in a general complex linear model system and in a Bayesian framework. We discuss formulating model selection problems and incorporating context-dependent {\it a priori} information through different levels of prior specifications. We also derive analytic Bayes factors and their approximations to facilitate model selection and discuss their theoretical and computational properties. We demonstrate our Bayesian approach based on an implemented Markov Chain Monte Carlo (MCMC) algorithm in simulations and a real data application of mapping tissue-specific eQTLs. Our novel results on Bayes factors provide a general framework to perform efficient model comparisons in complex linear model systems.
\end{abstract}

\section{Introduction}
Genetic association studies aim to detect statistical associations between genetic variants (most commonly, single nucleotide polymorphisms, or SNPs) and phenotypic traits. Genetic associations are complicated in nature: multiple SNPs may simultaneously affect a single phenotype, the genetic effects of a SNP with respect to a phenotype may exhibit a large degree of heterogeneity in different environmental conditions (known as gene-environment interactions), and a single SNP may affect multiple phenotypes through gene networks.
Statistical analysis of genetic associations under these complex settings has become increasingly important because it can yield a comprehensive understanding of the roles played by genetic variants in a biological system. To illustrate, we briefly introduce two motivating examples.

{\bf Motivating Example 1: Multiple-Tissue eQTL Mapping.}
eQTLs (expression quantitative trait loci) are genetic variants associated with gene expression phenotypes and play important roles in transcriptional regulation processes. Most recently, eQTL data have been collected from multiple tissue/cell types (e.g., the NIH GTEx project). One important goal is to identify eQTLs across tissues and investigate how their effects vary in different cellular environments. Biologically, it is expected that a proportion of eQTLs are active (i.e., effect size $\ne 0$) only in certain tissues but silent (i.e., effect size $= 0$) in others, a classic case of gene-environment interaction; for tissues in which an eQTL is active, the regulatory environments of the target gene are likely similar, and the effects of the eQTL are expected to show {\em low heterogeneity}.
In addition, because a single gene is typically subject to many regulatory elements, it is highly likely that there exist multiple eQTLs for any given gene.
Finally, in the most popular experimental design of this type, multiple tissue samples are collected from the same set of individuals, and intraindividual correlations of gene expressions need to be accounted for. Under this setting, it is challenging to simultaneously identify multiple and potentially tissue-specific eQTLs.
 
{\bf Motivating Example 2: Fine-Mapping in a Genetic Association Meta-Analysis.}
Genetic association studies with limited sample sizes are underpowered to detect modest association signals. Nevertheless, genuine genetic associations typically show {\em consistent} effect sizes in many independent studies.
Meta-analysis therefore becomes critically important to aggregate sample sizes and increase power for detecting associations.
Currently, most existing meta-analytic approaches in genome-wide association (GWA) studies analyze one SNP at a time. In a meta-analytic setting, the simultaneous mapping of multiple genetic associations, especially in a predefined genomic region, remains a statistical challenge.

Although identifying non-zero genetic associations can be naturally formulated as a model-selection problem, most available approaches (\cite{Fridley2009, Wilson2010,Wu2009,Mitchell1988,Guan2011}), applicable only to single multiple linear regression models, are inadequate for addressing the situations described in our motivating examples. This is mainly because, in both cases, observed data form subgroups (viz., different tissue types in eQTL mapping and individual GWA studies in the meta-analysis). We not only require a complex model system to account for these subgroup structures (in likelihood computation), but we also require variable selections to be performed either with respect to (as in the case of tissue-specific eQTLs) or integrating among (as in meta-analysis) the intrinsic subgroup structures. Furthermore, as we have shown in both examples, there typically exists {\it a priori} information on the correlations of non-zero effects. Effectively utilizing this prior information would greatly improve the performance of model selection and make the results easy to interpret.

In this paper, we describe a general system of linear models that is capable of addressing both of the motivating examples. We consider the problem of formulating model (variable) selection through prior specification under this linear system and propose Bayesian solutions to conduct model comparison and model selection via Bayes factors. We illustrate our Bayesian approach through simulation studies and a real example of tissue-specific eQTL mapping. We want to emphasize that our results on Bayes factors, discussed in section \ref{bf.rst.sec}, are completely general and  can be readily applied to a wide range of model comparison, hypothesis testing and model selection problems.

\section{A System of Simultaneous Multivariate Linear Regressions (SSMR)}\label{likelihood.section}
We describe a very general linear model system for which many commonly used linear models become special cases. It naturally applies in the complex scenarios in genetic association studies we have discussed. Unless otherwise specified, all of the results presented in this paper apply to this most general form of the linear model system. 

\subsection{Model Description and Notation}
 
We consider a system of simultaneous multivariate linear regressions (SSMR) consisting of a set of $s$ separate multivariate linear regression equations, i.e., 
\begin{equation}
  \Yv_i = \Xv_{c,i}\Bv_{c,i} + \Xv_{g,i} \Bv_{g,i} + \Ev_i,~ \Ev_i \sim {\rm MN} \left(\bf{0} ,\Iv, \Sigv_i \right), ~~ i= 1,\dots,s,
\end{equation}
where ``MN" denotes the matrix-variate normal distribution, and each composing linear equation describes one of the $s$ non-overlapping subgroups of observed data. For subgroup $i$ with $n_i$ subjects, $\Yv_i$ is an $n_i \times r$ matrix with each row representing $r$ quantitative measurements from one subject.
We denote $\Xv_i = (\Xv_{c,i} ~ \Xv_{g,i})$ as the $n_i \times (q_i+p)$ design matrix, in which $\Xv_{g,i}\, (n_i \times p) $ represents the data matrix of $p$ explanatory variables of interest (e.g., genotypes of interrogated genetic variants), and $\Xv_{c,i}\, (n_i \times q_i)$ represents the data of $q_i$ additional variables (including the intercept) to be controlled for;
matrices $\Bv_{g,i}$ ($p \times r$) and $\Bv_{c,i}$ ($q_i \times r$) contain the regression coefficients for the explanatory and the controlled variables, respectively.
Finally, $\Ev_i$ is an $n_i \times r$ matrix of residual errors in which each row vector is assumed to be independent and identically distributed as ${\rm N}({\bf 0}, \Sigma_i)$ (i.e., $\Ev_i \sim {\rm MN} \left(0 , \Iv, \Sigv_i \right)$).
Although the same set of $r$ response variables and $p$ explanatory variables are assumed to be measured in all $s$ subgroups, we allow each composing linear model to control for a different set of covariates. Furthermore, the residual errors are assumed to be independent across subgroups. 
In addition, we denote $\YYv:= \{\Yv_1,\dots,\Yv_s\}$, $\XXv := \{\Xv_1,\dots,\Xv_s\}$ and $\boldmath{\mathcal{E}} := \{\Sigv_1 ,\dots, \Sigv_s\}$. (Throughout the paper, we refer to $\SSigv$ as ``error variances".)

The SSMR model is a generalization of a class of linear systems; some commonly used special cases include the following:
\begin{enumerate}
  \item  Multiple Linear Regression: $s=1$ and $r=1$.
  \item  Multivariate Linear Regression (MVLR): $s=1$. This is a suitable model for describing multiple-tissue eQTLs for which different tissue samples are obtained from the same set of individuals (Motivating Example 1).
  \item  Systems of Simultaneous Linear Regressions (SSLR): $r=1$. This model can be applied to fine mappings of genetic variants in a meta-analytic setting (Motivating Example 2).
\end{enumerate}
The general SSMR model is also uniquely important for many genetics/genomics applications. One such example is the meta-analysis of genetic variants with respect to multiple phenotypes.

We introduce the vectorized regression coefficients $\bvg := {\tiny \left( \begin{array}{c} \vect(\Bv_{g,1}') \\ \vdots \\  \vect(\Bv_{g,s}') \\ \end{array}\right)}$ and $\bvc := {\tiny \left( \begin{array}{c} \vect(\Bv_{c,1}') \\ \vdots \\  \vect(\Bv_{c,s}') \\ \end{array}\right)}$, which are mathematically convenient to work with. We use the notation ${\lv}(\beta_{g,i})$ to denote an indicator function of the $i$-th component of $\bvg$, such that ${\lv}(\beta_{g,i}) = 1$ if $\beta_{g,i} \ne 0$ and 0 otherwise. Furthermore, we define the following indicator vector:
\begin{equation}
  \label{skeleton}
  \xi(\bvg) := \left({\lv}(\beta_{g,1} ), {\lv}(\beta_{g,2}), \dots \right).
\end{equation}
In this paper, $\xi(\bvg)$ is our quantity of interest for model selection.

To perform Bayesian inference based on the SSMR model, we assign prior distributions for $\bvg, \bvc$, and $\Sigv$.
For $\bvg$, we assume a multivariate normal prior,
\begin{equation}
   \label{prior.wg}
  \bvg \sim {\rm N}({\bf 0}, \Wv_g).
\end{equation}	
The variance-covariance matrix $\Wv_g$ plays a central role in our framework, and we defer a detailed discussion of it to section \ref{prior.section}.
For the regression coefficients of controlled variables, we assume
\begin{equation}
    	\bvc \sim{\rm N} \left({\bf 0}~,~\Psiv_c \right),
\end{equation}
where matrix $\Psiv_c$ is assumed to be diagonal. When performing an inference, we consider the limiting condition $\Psiv_c^{-1} \to \bf {0}$ (i.e., each composing coefficient in $\bvc$ is effectively assigned an independent flat prior). Furthermore, we assume $\bvg$ and $\bvc$ are {\it a priori} independent.
Finally, we assign an independent inverse Wishart prior, with parameters $m_i$ (a positive scalar) and $\Hv_i$ (a positive-definite $r \times r$ matrix), for each composing $\Sigv_i \in \SSigv$, i.e.,
 \begin{equation}
 	\label{prior.sigv}
	\Sigv_i \sim {\rm IW}_r(\nu_i \Hv_i, m_i),
 \end{equation}
where $\nu_i = m_i - q_i -r -1$, and we require $\nu_i > 0$. If $r$ is small relative to the sample size, $\Sigv_i$ can be sufficiently learned from the data. In such cases (as in the simulations and the data application of this paper), we consider the limiting condition $\Hv_i \to 0$ and $\nu_i \to 0$. As  $r$ is large, setting $\Hv_i$ and $\nu_i$ requires context-dependent considerations, we discuss this briefly in the discussion.

\section{Prior Specification for Structured Model Selection in SSMR}\label{prior.section}
At its most basic level, a model/variable selection problem in the SSMR model can be formulated as an inference on $\xi(\bvg$) (defined in Equation(\ref{skeleton})). Throughout this paper, we refer to a candidate model as a particular configuration of $\xi(\bvg)$. In our Bayesian framework, a prior distribution on the space of candidate models, $\pxi$, is used to prioritize (or in the extreme case, enforce) a certain class of preferred models. 
For instance, the intrinsic (sub)group structure and the sparse property of preferred candidate models can be quantified by $\pxi$.
{\em Given a candidate model}, we use the multivariate normal prior (\ref{prior.wg}) to fully specify the prior distribution on $\bvg$, for which a positive semidefinite covariance matrix $\Wv_g$ is sufficient. In this presentation, we use matrix $\Wv_g$ to serve two primary purposes:
\begin{enumerate}
\item articulate the structure of the given candidate model $\xi(\bvg)$.
\item convey context-dependent {\it a priori} correlation information on non-zero elements of $\bvg$ to aid model selection.
\end{enumerate}
The first point provides convenience in mathematical representations, and the second point highlights the fact that matrix $\Wv_g$ incorporates a source of prior information that complements what is conveyed in $\pxi$.

The idea of using matrix $\Wv_g$ to represent a candidate model is similar to the use of the ``spike-and-slab" prior in Bayesian variable selection: for a regression coefficient $\beta \in \bvg$, it is convenient to represent $\Pr(\beta = 0) = 1$ by a degenerate normal (prior) distribution $\beta \sim {\rm N}(0,0)$ (i.e., a spike), and accordingly, a non-zero marginal prior variance on $\beta$ (i.e., a slab) indicates the corresponding variable is included. Thus, information about $\xi(\bvg)$ can be directly obtained from the main diagonal of a given (singular) matrix $\Wv_g$.

The off-diagonal of matrix $\Wv_g$ defines {\it context-dependent} prior correlations between (non-zero) regression coefficients.
Incorporating this information in the inference enables ``borrowing strength" across correlated components in $\bvg$, thereby improving the efficiency of model selection.
Given a specific context and a candidate model, the qualitative dependence relationships between any two coefficients in $\bvg$ are typically determined. Much recent research has been devoted to further quantifying such correlation structures (\cite{Scott-Boyer2012, Guan2011,Wen2011}. We provide a brief summary of some existing prior specification approaches in various genetic settings in Appendix A.

\subsection{Parameterization of $\Wv_g$ for Model Selection}
To better facilitate model selection, we propose to parameterize $\Wv_g = (\Gammav_g, \Lambda_g)$, where $\Gammav_g$ is a binary matrix consisting of entry-wise non-zero indicators and is identical in size and layout to $\Wv_g$; $\Lambda_g= \{w_{ij}\}$ is an indexed set of numerical values quantifying each non-zero entry in the $\Gammav_g$ matrix.
For a given candidate model, the main diagonal of $\Gammav_g$ corresponds to $\xi(\bvg)$. The off-diagonal of $\Gammav_g$ represents the qualitative prior dependence relationships between coefficients in $\bvg$ and can always be deterministically specified given its diagonal and a specific application context.
Mathematically speaking, there always exists a context-dependent injection from $\xi(\bvg)$ to $\Gammav_g$.

Given the prior probability $\pxi$ for a candidate model, we now have a principled way to specify a prior distribution on matrix $\Wv_g$, i.e.,
\begin{equation}
  \label{wg.prior}
  p(\Wv_g) =  p(\Lambda_g \mid \xi(\bvg))\cdot \pxi.
\end{equation}

\subsection{Scale-Invariant Prior Formulation}
In practice, it is often desirable that inference results be invariant to linear transformations of response variables (the $g$-prior for multiple linear regressions and the conjugate prior commonly used in the MVLR model both have this property, see also \cite{Servin2007, Wen2011}). To achieve this in the SSMR model, we scale each element in $\bvg$ by its corresponding marginal residual standard error (in the MVLR, the residual standard error for a given regression coefficient is represented as the square root of the corresponding diagonal element in its residual variance-covariance matrix).
More formally, we define a vector of scale-free standardized effects by
$\sbv_g := \SSv^{-\frac{1}{2}} \bvg$, where $\SSv$ is a diagonal matrix permuted from $\oplus_{i=1}^s \left(\Iv \otimes{\rm diag}(\Sigv_i)\right)$ to match the order of elements in $\bvg$. (Throughout this paper, we use``$\otimes$" and ``$\oplus$" to denote Kronecker product and direct sum of matrices, respectively). Under this setting, a multivariate normal prior distribution $\sbv_g \sim {\rm N}({\bf 0}, \Uv_g)$ induces a normal prior distribution on $\bvg$ with mean 0 and
\begin{equation}
  \label{unit.free.trans}
  \Wv_g = \SSv^{\frac{1}{2}} \,\Uv_g\, \SSv^{\frac{1}{2}}. 
\end{equation}
With (\ref{unit.free.trans}), we are able to handle the desired scale-invariant prior formulation as a special case of the original scale formulation.

\section{Results on Bayes Factors}\label{bf.rst.sec}
We derive Bayes factors to facilitate model comparisons and selections in the SSMR model. At the most fundamental level, Bayes factors enable us to compare the supporting evidence from observed data for a set of competing models (which are not necessarily nested). In the case that posterior model probabilities are of direct interest,  
Bayes factors can typically be utilized as computational devices in the place of marginal likelihood, which is sometimes more difficult to compute. In what follows, we discuss the Bayes factors derived from the SSMR model, assuming the multivariate normal prior (\ref{prior.wg}) is fully specified. 
Let $H_0$ denote the trivial null model, where $\bvg \equiv 0$. Then, for an alternative target model characterized by its prior variance $\Wv_g$, we formally define a null-based Bayes factor (\cite{Liang2008}) as follows: 
\begin{defn}
{\it
 Under the SSMR model, for a positive definite $\Wv_g$, the Bayes factor is defined as
 \begin{equation}
 \label{pd.wg.bf}
  \BF(\Wv_g) =  \lim_{\Psiv_c^{-1} \to 0}  \frac{P(\YYv|\XXv,\Wv_g)}{P(\YYv|\XXv, H_0)}.
  \end{equation}
}
\end{defn}  
For technical reasons, the above definition requires $\Wv_g$ to be full rank; we will extend this definition to allow for a singular $\Wv_g$ matrix later in section \ref{singular.ex}.

\subsection{Analytic Results of Bayes Factors}
We start by introducing some necessary additional notation. We use $\hbvg$ to denote the maximum likelihood estimate (MLE) of $\bvg$ and denote its variance by $\Vv_g := {\rm Var}(\hbvg)$. Under the SSMR model, both $\hbvg$ and $\Vv_g$ have closed-form expressions: $\hbvg$ depends only on observed data $\XXv$ and $\YYv$, while $\Vv_g$ depends on $\XXv$ and $\SSigv$ (their explicit functional forms can be found in Appendix B).

\subsubsection{Exact Bayes Factors with Known Error Variances}
In the general case of the SSMR model, when the error variances are considered known, rather than being assigned priors, the exact Bayes factor can be analytically expressed. We summarize this result in the following lemma:
\begin{lemma}
In the SSMR model, if $\SSigv$ is known, the Bayes factor in definition 1 can be analytically computed by
\begin{equation}
  \label{exact.bf}
  {\rm BF}(\Wv_g) = | \Iv + \Vv_g^{-1}\Wv_g|^{-\frac{1}{2}} \cdot \exp \left( \frac{1}{2} \hbvg' \Vv_g^{-1} \left[\Wv_g (\Iv +\Vv_g^{-1}\Wv_g)^{-1}\right] \Vv_g^{-1} \hbvg \right).
\end{equation}
\end{lemma}
The derivation of Lemma 1 is mostly straightforward; the details are provided in Appendix B.1.
 
\begin{note} {\it The Bayes factor naturally addresses potential collinearity in predictors. In particular, the evaluation of the Bayes factor does not require the involved design matrices to be full rank (the details are explained in Appendix C). As a result, when highly correlated explanatory variables are included in the model, the Bayes factor can still be stably computed without special computational treatments.}
\end{note}
Note 1 is extremely relevant for genetic applications, where genotypes of many spatially close genetic variants are often highly correlated.  
 
\subsubsection{Approximate Bayes Factors with Unknown Error Variances}
In more realistic settings, error variances are typically unknown
and additional integrations with respect to $\SSigv$ are necessary for Bayes factor evaluations. Except for a very few special cases, the exact Bayes factor generally is analytically intractable.
Alternatively, we apply Laplace's method to pursue analytic approximations of the Bayes factor.
Laplace's method has been widely applied in computing Bayes factors in other similar settings (\cite{Kass1995, Raftery1996, Diciccio1997,Saville2009, Wen2011}). In the case of the SSMR model, applying Laplace's method yields an analytic approximation that maintains the exact functional form of (\ref{exact.bf}) -- only with the unknown $\Sigv$ replaced by an intuitive point estimate.
More specifically, $\ABF$ substitutes each $\Sigv_i \in \SSigv$ in  (\ref{exact.bf}) with the following Bayesian shrinkage estimate
\begin{equation}
 \label{bayes.Sigma.est}
   \check \Sigv_i = \frac{\nu_i}{n_i + \nu_i} \Hv_i + \frac{n_i}{n_i + \nu_i} \left[\alpha_i \hat \Sigv_i + (1-\alpha_i) \tilde \Sigv_i \right],
\end{equation}
where $\hat \Sigv_i$ and $\tilde \Sigv_i$ denote the MLEs of error variances estimated from the residuals under the target and the null models, respectively, parameters $\nu_i$ and $\Hv_i$ are defined in the inverse-Wishart prior of $\Sigv_i$, and parameter $\alpha_i \in [0,1]$ serves as a tuning parameter and has an impact on the finite-sample accuracy of the resulting Bayes factor approximations. We further denote $\av = ( \alpha_1,\dots,\alpha_s) $ and $\hSigv := \{\check \Sigv_1,\dots,\check \Sigv_s\}$.

Other relevant quantities in (\ref{exact.bf}) that are functionally related to $\SSigv$ include $\Vv_g$ and potentially $\Wv_g$ (e.g., in the scale-invariant prior formulation). We denote $\check \Vv_g$ and $\check \Wv_g$ as the corresponding plug-in estimates of $\Vv_g$ and $\Wv_g$ by $\hSigv$. 

The result of the approximate Bayes factor is summarized in the following proposition:
\begin{prop}{\it
Under the SSMR model, when $\SSigv$ is unknown, applying Laplace's method leads to the following analytic approximation of the Bayes factor
\begin{equation}
   \label{abf1.general}
   {\rm ABF}(\Wv_g, \av) := |\Iv + \check \Vv_g^{-1}\check \Wv_g |^{-\frac{1}{2}} \cdot \exp \left( \frac{1}{2} \hbvg' \check \Vv_g^{-1}\left[ \check \Wv_g (\Iv + \check \Vv_g^{-1} \check \Wv_g )^{-1} \right]\check \Vv_g^{-1} \hbvg \right).
\end{equation}
It follows that
$$  \BF(\Wv_g) = \ABF(\Wv_g, \av)\cdot\prod_{i=1}^{s}\left(1+O(n_i^{-1})\right). $$
}
\end{prop}
\begin{proof}
 See derivation in Appendix B.2.
\end{proof}
As long as $\av$ resides in an $s$-simplex, the above proposition holds. There are two notable extreme cases concerning the choice of $\av$ values:
\begin{enumerate}
 \item $\alpha_1=\cdots=\alpha_s = 1$. The resulting $\hSigv$ only relates to the MLEs estimated from the target model, i.e., $\check \Sigv_i = \frac{\nu_i}{n_i + \nu_i} \Hv_i + \frac{n_i}{n_i + \nu_i} \hat \Sigv_i$.
Under the usual asymptotic settings, where $n_i \gg p$ and $n_i \gg r$ and when the mean model is correctly specified, $\check \Sigv_i \overset{a.s.}\rightarrow \Sigv_i$. By the continuous mapping theorem, it follows that the resulting $\ABF$ almost surely converges to the true value.
 \item $\alpha_1=\cdots=\alpha_s = 0$. $\check \Sigv_i$ only relies on the MLE of $\Sigv_i$ estimated from the trivial null model, i.e., $\check \Sigv_i = \frac{\nu_i}{n_i + \nu_i} \Hv_i + \frac{n_i}{n_i + \nu_i} \tilde \Sigv_i$.
Indeed, $\hbvg$ can also be analytically expressed as a simple analytic function of the MLEs of the regression coefficients obtained from the null model. As a result, computing this particular $\ABF$ only requires fitting the trivial null model -- a scenario analogous to computing score statistics in hypothesis testing (the details are further explained in Appendix F.1).
\end{enumerate}
Notwithstanding their having the same asymptotic order of error bounds, different $\av$ values affect the accuracy of the approximations in finite-sample situations.
To examine the performance of $\ABF$s with various $\av$ values, we carry out numerical experiments with small sample sizes. In summary, we find that the resulting $\ABF$ with all $\alpha_i = 1$ tends to be anti-conservative compared with true values (most likely because $\hat \Sigv_i$ is prone to overfitting in these cases), whereas setting all $\alpha_i = 0$ understandably yields conservative approximations. Interestingly, setting $\alpha_i = 0.5$ for all subgroups gives consistently accurate numerical results in our simulation setting. Finally, we confirm that as sample sizes grow, all approximations become increasingly accurate, regardless of $\av$ values.
The details of the numerical comparisons and the results are given in Appendix E.

\subsubsection{Singular Prior Distributions}\label{singular.ex}
To extend the definition of Bayes factors for a singular $\Wv_g$, we first define
\begin{equation}
  \label{wg.limit.def}
 \Wv_g^\dagger(\lambda) = \Wv_g + \lambda \Iv, ~~\lambda > 0,
\end{equation}
where $\Wv_g$ is only required to be positive semidefinite. We then are able to extend definition 1 to include a singular $\Wv_g$ matrix:
\begin{defn} {\it
  Under the SSMR model, for a positive semidefinite $\Wv_g$, the Bayes factor is defined as
  \begin{equation}
   \label{singular.wg.bf}
  {\rm BF}(\Wv_g) = \lim_{\lambda \to 0} {\rm BF}\left(\Wv_g^\dagger(\lambda)\right).
  \end{equation}
  }
\end{defn}  
This definition is based on the following important intuition: Bayes factors are expected to vary very smoothly over a continuum of models. This is not only desirable but also critically important for selecting models consistently when using Bayes factors.
We obtain the following result regarding the existence of the limits:
\begin{prop}{\it
  For the SSMR model, the limiting Bayes factors in definition 2 are always well defined, provided that $\Wv_g$ is positive semidefinite. } 
\end{prop}
\begin{proof}
 See Appendix D.
\end{proof}
Proposition 2 directly extends the results of Lemma 1 and Proposition 1 to allow for a singular $\Wv_g$ matrix. Moreover, when approximating Bayes factors using Laplace's method, the functional form of the result remains the same; however, we now compute the MLE of the unknown $\Sigv_i$ for the target model, subject to the linear restrictions imposed by the singular $\Wv_g$ matrix. The details are explained in Appendix D.

\subsection{Connections to Frequentist Test Statistics and the BIC}

Previous studies by \cite{Wakefield2009, Johnson2005, Johnson2008, Wen2011} have shown in certain linear model systems (all being regarded as special cases of the SSMR model) that Bayes factors are linked to commonly used frequentist test statistics. We also identify approximate Bayes factors for the SSMR model as being connected to the multivariate Wald statistic and Rao's score statistic, depending on the choice of $\av$ value. The main consequence of this connection is that under specific prior specifications of $\Wv_g$, Bayes factors and the corresponding test statistics yield the same ranking for a set of models.

Bayes factors are also naturally linked to the Bayesian Information Criterion (BIC, \cite{Schwarz1978}). Under the SSMR model, we show (in Appendix F) that the BIC can be derived as a very rough (i.e., with error bound $O(1)\,$ in log scale) approximation to both the exact and the approximate Bayes factors for most $\Wv_g$ matrices. Because the BIC is known to be asymptotically consistent as a model selection criterion, based on this connection, we conclude that our Bayes factors also enjoy this property.

A detailed explanation of both connections is given in Appendix F.

%

\subsection{Bayes Factors of Candidate Models}

Based on the results of $\BF(\Wv_g)$ and Equation (\ref{wg.prior}), we can compute the Bayes factor of a given candidate model, $\xi(\bvg)$, by
  \begin{equation}
  \label{bf.skeleton}
  \BF\left(\xi(\bvg)\right) =  \int p\left(\Lambda_g \mid \xi(\bvg)\right)\, \BF(\Wv_g)\, d\,\Lambda_g,
\end{equation}
which essentially integrates out the effect sizes of non-zero regression coefficients.
In many genetic applications, it is feasible and effective to model $p(\Lambda_g \mid \xi(\bvg))$ by a finite discrete distribution (\cite{Servin2007, Stephens2009, Wen2011}). In these cases, the integration in (\ref{bf.skeleton}) is replaced by a summation, and the computation is efficient.

\section {Bayesian Model Selection Procedure and the MCMC Algorithm}\label{mcmc.section}

Based on the results discussed in the previous sections, we are now ready to describe the full Bayesian model selection procedure based on the SSMR model. Assuming the goal of inference is $\xi(\bvg)$, the following prior information is required to be specified in a {\em context-specific} manner:
\begin{enumerate}
  \item prior distribution in the space of candidate models, $\pxi$.
  \item injection from $\xi(\bvg)$ to $\Gammav_g$, i.e., specification of prior qualitative dependence/independence structures.
  \item probability distribution $p(\Lambda_g \mid \xi(\bvg))$, i.e., quantification of prior correlation and marginal variance specified in $\Gammav_g$.
\end{enumerate}
Then, based on Equation (\ref{bf.skeleton}) and relevant discussions on Bayes factor computations, it is straightforward to perform full Bayesian model selection under the general SSMR model. If the number of the candidate models, $2^{rps}$ in total, is computationally manageable, we can enumerate all possible models and evaluate their posterior probabilities directly. However, in most practical settings, the candidate model space is  enormous, we then need the MCMC algorithm to efficiently traverse the model space. 

For the sake of simplicity but without loss of generality, we give a detailed description of a particular version of this algorithm for the commonly used MVLR model in Appendix H. Aided by a novel proposal distribution proposed by \cite{Guan2011}, we observe that the implemented Markov chain achieves fast mixing and generates accurate results even in very high-dimensional settings. The performance of the algorithm is demonstrated through simulations and real data applications in sections \ref{sim.section} and \ref{app.section}.

\section{Simulation Studies}\label{sim.section}

We perform simulation studies to examine and demonstrate the performance of the proposed Bayesian methods in a variety of settings. In these simulations, we focus on the scenario of mapping eQTLs across a handful of tissue types using a common set of individuals, which is best described by an MVLR model with large $p$ (number of candidate genetic variants), small $n$ (sample size), and small $r$ (number of tissue types) values. Moreover, we allow each covariate (SNP) to have different (zero or non-zero) effects in $r$ subgroups (tissues); however, within a covariate, we simulate a scenario in which non-zero effects across subgroups are highly correlated.
 
\subsection{Simulation Settings}

We create two simulation settings that differ in the generation of covariates. In the first setting, we simulate $p=250$ {\em independent} covariates for $n=100$ unrelated individuals. The causal SNPs (i.e., the covariates that are associated with the phenotype in at least one of the $r$ subgroups) are independently assigned by a ${\rm Bernoulli}(0.03)$ distribution.
In the second setting, we focus on correlated covariate data. More specifically, we take real SNP genotype data from 100 Caucasian samples of the 1000 Genomes project. We select 105 genomic regions across chromosome 22 that average 30 kb in size. The two consecutive regions are approximately 300 kb apart, and within each region, we select 15 SNPs whose minor allele frequencies are greater than 5\%. Between and within these genomic regions, the genotypes present various degrees of spatial correlations (also known as linkage disequilibrium, or LD). In this setting, the regions harboring causal SNPs are assigned by a ${\rm Bernoulli}(0.03)$ distribution, and we randomly assign a single causal SNP within the selected region.

Given the covariate data, we simulate quantitative (gene expression) phenotype data in $r=3$ subgroups (tissue types) for each individual using the following scheme. For each SNP, we represent its binary association states by an $r$-vector (e.g., $\gav = (100)$ indicates a causal and tissue-specific eQTL for which association only presents in the first tissue type), and collectively, $\{\gav_i: i=1,...,p\}$ represents the true $\xi(\bvg)$. We randomly assign each causal SNP a non-zero configuration according to a discrete distribution. More specifically, among seven possible non-zero configurations, $\gav = (111)$ is assigned a probability of 0.50, and the others are assumed equally likely (i.e., with probability $1/12$ each), conditioning on $\gav \ne (000)$. This distribution is motivated by the observation from the real multiple-tissue eQTL data, where most identified eQTLs are found to have consistent effects in all tissues. For each simulated $\gav \ne (000)$, we first generate a mean effect from $\bar \beta \sim {\rm N}(0,1)$; then, non-zero genetic effects are subsequently drawn from $\beta \sim {\rm N}(\bar \beta,  \frac{{\bar \beta}^2}{100})$. With this procedure, the non-zero $\beta$s for a causal SNP across tissues are highly correlated, albeit with some non-negligible heterogeneity. Finally, the residual errors for each individual are independently simulated from a multivariate normal distribution, $\ev \sim {\rm N}(0, \Sigma)$, with $\Sigma = {\tiny \left( \begin{array}{ccc} 1.00 & 0.24 & 1.20 \\  0.24 & 1.44 & 1.08 \\ 1.20 & 1.08 & 2.25 \end{array}\right) }$ prefixed. We generate 200 and 500 data sets for simulated independent and real correlated genotypes, respectively.

\subsection{Bayesian Model Selection} \label{bayes.mvlr.sec}

We perform inference on the binary indicator vector $\xi(\bvg)$. We assume that genetic effects are {\it a priori} independent across SNPs but correlated among tissues within a single covariate if they are non-zero. This prior relationship is precisely formulated by an injection: $\Gammav_g = \oplus_{i=1}^p [\gav_i \otimes \gav_i']$, and the factorization of prior probability, $\pxi = \prod_{i=1}^p \Pr(\gamma_i)$.

In all cases, we assume the default prior probability $\Pr \left (\gav = (000)\right) = 0.99$ for each covariate, which encourages an overall sparse $\xi(\bvg)$. By default, all possible non-zero configurations for $\gav$ are assigned with equal prior probability, $0.01 \times \frac{1}{2^r-1}$.

To specify the distribution $\Lambda_g \mid \xi(\bvg)$, we follow \cite{Wen2011, Flutre2012} and model the joint prior distribution of a pair of non-zero effects within a covariate by a multivariate normal ${\tiny \left(\begin{array}{c} \beta_1 \\ \beta_2 \\ \end{array}\right) \sim {\rm N} \left[ 0~,~  \left( \begin{array}{cc} \omega^2 + \phi^2 & \omega^2 \\ \omega^2  &  \omega^2 + \phi^2  \end{array}\right)  \right]}$, where parameter $\phi$ describes the prior heterogeneity of the effects, and parameter $\omega$ characterizes the magnitude of the average prior effect, and the prior correlation between the pair can be computed by $\omega^2/(\omega^2+\phi^2)$ (details explained in Appendix A.3). 
Furthermore, instead of fixing a single $(\phi, \omega)$ value for all covariates, we assume that $(\phi_{i}, \omega_{i})$ for covariate $i$ is independently and uniformly drawn from the following set \newline
$~~~~~~~~~L := \{ (\phi^{(l)}, \omega^{(l)}): (0.05,0.20), (0.10,0.40), (0.20, 0.80), (0.40,1.60)\},$ \newline
where the various levels of $\omega$ values cover a range of potentially small, modest, and large average effects and the relatively small $\phi$ value quantifies our prior belief of low heterogeneity across non-zero effects. It is worth emphasizing that even with a single grid value, the prior would allow for a range of actual effect sizes, and multiple grid points (which form a mixture normal distribution) are helpful for describing a longer-tailed distribution of effect size. It should also be noted that all the priors we use in the inference are different from the true generative distributions used in the simulations.

For likelihood calculation of a given $\xi(\bvg)$, we compute a Bayes factor (\ref{bf.skeleton}) in which $\BF(\Wv_g)$ is approximated by $\ABF(\Wv_g, \alpha=0.5)$. We use the MCMC algorithm described in section \ref{mcmc.section} to conduct posterior inference.

For simulated independent genotype data, we use the posterior inclusion probability of each SNP configuration to assess its relative importance.
In the case of correlated covariate data, it might not be plausible to identify the true association based on observed data (e.g., in a scenario in which multiple covariates are perfectly correlated). Therefore, we focus on assessing the importance of preselected genomic regions and compute the posterior probability that a given region harbors a genetic variant with particular configurations. These quantities are computed by combining SNP-level posterior inclusion probabilities and posterior model probabilities using the inclusion-exclusion principle.

\subsection{Methods for Comparison}
We compare our Bayesian model selection method (BMS) with two other methods: single variable analysis, which examines one covariate at a time while accounting for the subgroup structure (\cite{Wen2011, Flutre2012}), and the regularized regression approach LASSO (\cite{Tibishirani1996}).

The single variable procedure can be viewed as a special case of the general MVLR model with $p=1$. For each SNP, we compute the single-SNP posterior probability of each configuration based on the corresponding $\ABF$ values and use it to assess the importance of each SNP configuration. For the real genotype data, we analyze one region at a time and further compute a regional posterior probability based on single-SNP Bayes factors, assuming at most one causal SNP in a region, a method described in (\cite{Servin2007, Flutre2012}).

We center the phenotype data and apply the LASSO procedure to estimate $\bvg$ by
\begin{equation}
  \label{lasso.eqn}
  \arg \min_{\bvg} \bigg( \vect(\Yv) - (\Iv \otimes \Xv) \bvg \bigg)' \bigg( \vect(\Yv) - (\Iv \otimes \Xv) \bvg \bigg)+ \lambda\sum_j |\beta_{g,j}|,
\end{equation}
where $\Iv$ is the $r \times r$ identity matrix and $\lambda$ is the tuning shrinkage parameter. If $\lambda$ is sufficiently large, LASSO produces sparse estimates of $\bvg$; whereas, if $\lambda$ is set to 0, the solution becomes the usual least squares estimate/MLE for the MVLR model.
Given a particular $\lambda$ value, for the simulated independent genotype data, we identify the true and false positives of non-zero $\bvg$ estimates; whereas, for the real genotype data, following \cite{Guan2011}, we further denote that a region is positively identified if any SNP within that region is selected by LASSO. 
We then record the full solution paths from LASSO for a range of $\lambda$ values using the lars package (version 1.1) implemented in R.

\subsection {Simulation Results}

In both simulation settings, we represent the results in Figure \ref{sim1.figure} by plotting curves of the trade-off between true and false positives from all three experimental methods. Each point on the curve is obtained by accumulating true and false positives across independent simulated data sets using a common threshold (either of the posterior inclusion probability or the shrinkage tuning parameter) within a method.
In both simulation settings, the Bayesian model selection method (BMS) always yields as many or more true positives than the other comparable methods for any given false-positive value.

Many previous publications have reported that multivariate methods are superior to single-variable analysis in selecting candidate variables in multiple linear regression models. We observe that a similar pattern also holds for multivariate linear regressions in our simulation settings. \cite{Guan2011} provide some very intuitive explanations for the superiority of multivariate methods vs. single-variable methods, even when covariates are all independent. Their arguments also naturally apply in our context. Although this result is largely expected, it serves as a reassuring sanity check that our implementation of the MCMC algorithm is fast mixing in this nontrivial setting (one could expect that a poor-mixing Markov chain would yield results inferior to those obtained from a single-variable analysis).

We conduct additional simulations to investigate the performance difference between BMS and LASSO.
First, we observe that the accuracy of LASSO is affected by correlated error structures characterized by $\Sigv$, which is not accounted for in (\ref{lasso.eqn}). Similar observations also have been made by \cite{Rothman2010}.
Second and more importantly, BMS utilizes additional correlation information on effect sizes within a single covariate through priors, whereas LASSO does not. We provide 
the details of these additional simulations and their results in Appendix I.

Finally, we notice that BMS performs in a stable manner even when covariate data are (highly) correlated, while LASSO greatly underperforms in such a setting.
\begin{figure}[h!t]
\begin{center}
\includegraphics[totalheight=0.42\textheight]{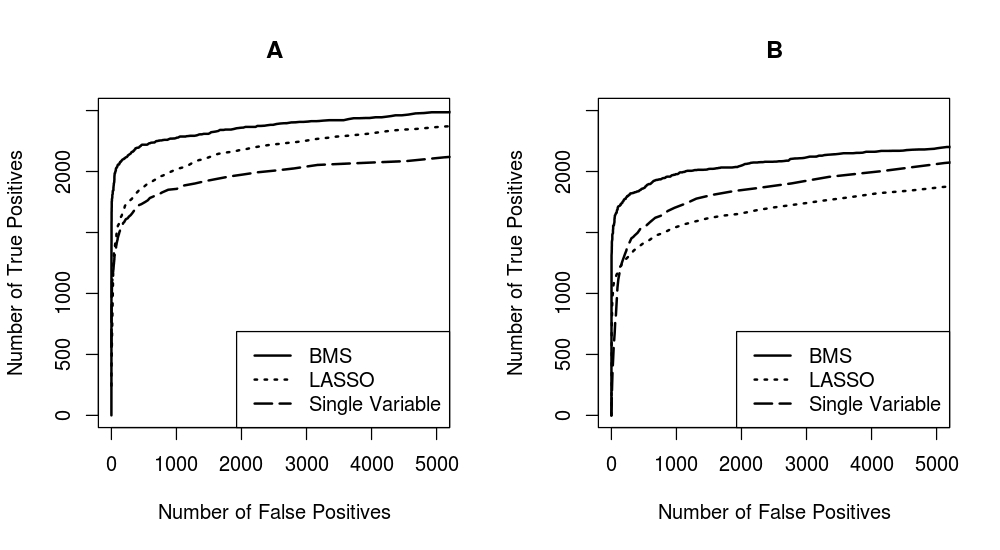}
\caption{\label{sim1.figure}Plots of the trade-offs between true positives and false positives for all three compared methods in two simulation settings. Panel A is based on simulated independent covariate data, and Panel B shows the results for correlated covariate data taken from real genotypes. In both cases, the proposed Bayesian model selection method (BMS) achieves superior performance. LASSO seems to severely underperform when covariates are correlated. }
 \end{center}
\end{figure}	
 
\section{Real Data Application}\label{app.section}

We apply the Bayesian model selection method to map eQTLs across multiple tissues on a real data set originally published by \cite{Dimas2009}. In this experiment, the investigators genotyped 75 unrelated western European individuals. Expression levels from this set of individuals were measured genome-wide in primary fibroblasts, Epstein-Barr virus-immortalized B cells (LCLs), and T cells. The expression data went through quality control and normalization steps by the original authors, and we further select a subset of 5,011 genes that are highly expressed in all 3 cell types and perform additional quantile-normalizations for each gene in each cell type. For demonstration purposes, we map eQTLs for each gene separately and narrow the search for eQTLs in the {\it cis}-region (i.e., the coding region and its close neighborhood) of each gene (note, this is also the strategy adopted in the original publication).

The setting of this data set is similar to that of our simulations. We use the MVLR model described in section \ref{bayes.mvlr.sec} to jointly infer the association states of all {\rm cis}-SNPs in three cell types for each selected gene. More specifically, we assume
the following independent priors for each SNP: $\Pr \left(\gav = (000)\right) = 0.99, ~\Pr\left(\gav = (111)\right) = 0.01 \times \frac{1}{2}$, and the rest of the six possible tissue-specific configurations are assigned probability mass $(0.01 \times \frac{1}{2} \times \frac{1}{6})$ each. This prior setup reflects our prior beliefs that the vast majority of {\rm cis}-SNPs are not eQTLs and that among eQTLs, most are likely to behave in a tissue-consistent manner. Finally, we use the same prior distribution of $p\left(\Lambda_g \mid \xi(\bvg)\right)$ described in section \ref{bayes.mvlr.sec}.

\noindent {\bf Remark 1}. It is important to note that genome-wide expression-genotype data are typically informative about the distributions of configurations of $\gav$ and effect-size grids in $L$. In other words, those distributional parameters can be effectively estimated by pooling information across all genes using a hierarchical model approach (\cite{Veyrieras2008,  Flutre2012}). In fact, the hyperparameters we select here are closely related to the estimations from fitting such a hierarchical model; however, these details are not our focus in this paper.

We apply the MCMC algorithm described in section \ref{bayes.mvlr.sec} to the set of 5,011 selected genes. We identify 510 ``eQTL genes" whose inferred best posterior model contains at least one candidate {\it cis}-SNP. In total, 539 eQTLs are identified from this set of the best posterior models, and 382 are inferred as tissue-consistent. 
Using the posterior maximum probability models, we are also able to confidently identify 28 genes with multiple {\it cis}-eQTLs accounting for linkage disequilibrium (LD), suggesting the involvement of multiple regulatory elements in transcriptional regulation processes.

One of the unique advantages of our Bayesian method is its ability to perform fine mapping on interesting genomic regions harboring true causal eQTLs. We demonstrate this feature through the analysis of gene C21orf57 (HGNC symbol YBEY, Ensemble ID ENSG00000182362). From a total of 236 {\it cis}-SNPs, our Bayesian analysis identifies three genomic regions centered around SNPs rs12329865 and rs2839265 and a SNP pair in perfect LD (rs2839156, rs2075906). The best posterior models consist of one SNP from each region, and the three regions have marginal posterior inclusion probabilities of 0.66, 0.38, and 0.89, respectively. More interestingly, our results suggest that the three distinct eQTL regions have completely different tissue activity configurations. We summarize these results in Table \ref{YBEY.tbl}. We further examine the effect sizes of the identified signals in each cell type separately, and the results (shown in appendix J) are strongly consistent with the conclusions of our tissue specificity inference.

As a comparison, we also applied the remMap method (\cite{Peng2009}, R implementation version 0.10) to the genotype-expression data of the gene C21orf57. The remMap method implements a penalized multivariate regression algorithm which assumes the same MVLR model. 
There are two tuning parameters required by the remMap method: one controls the sparsity of $\xi(\bv)$ and the other controls the sparsity of the residual error variance matrix. These two parameters are selected using a BIC procedure implemented in the R package.
In the end, remMap does not select any eQTLs. Given the strength of the signals identified by the Bayesian procedure and the results from the single SNP analysis, this is a little surprising. Nevertheless, we noted in a similar context of mapping eQTL for mutiple genes, \cite{Scott-Boyer2012} also observed this overly conservative behavior of the remMap method. We suspect that the non-trivial LD patterns presented in the SNP data might be one of the contributing factors here. As \cite{Peng2009} noted, complex correlation structures in predictors lead to the remMap procedure selecting very small models. In addition, like the LASSO procedure, the remMap method does not utilize the correlation information on eQTL effect sizes across tissues.

\begin{center}
\begin{table}[ht]
\begin{tabular} { c | c c c }
SNP & Position & Configuration & Posterior inclusion prob. \\
\hline
\hline
rs12329865 & chr 21:47583506 & LCL only & 0.662 \\
\hline
rs2075906   & chr 21:47625544  & consistent & 0.447\\
rs2839156   & chr 21:47641196  & consistent &  0.444\\
\hline
rs2839265   & chr 21:47867318   & Fibroblast only & 0.378 \\
\hline 
\end{tabular}
\vskip 0.8cm
\caption{\label{YBEY.tbl}Potential eQTLs identified by the Bayesian model selection procedure using only genotyped SNPs. Genotypes of SNPs rs2075906 and rs2839156 are highly correlated. The two models [rs12329865,rs2075906,rs2839265] and [rs12329865,rs2839156,rs2839265] have the highest posterior model probabilities (0.200 and 0.204, respectively) .}
\end{table}
\end{center}

To refine the identified genomic regions and rule out potential spurious associations identified with low-density SNPs, we perform genotype imputation to obtain additional genotypes of untyped SNPs using the 1000 Genome European panel and software package IMPUTE v2 (\cite{Howie2009}).  In the end, we accumulate genotypes from 4797 SNPs, roughly a 20-fold increase, for the same {\it cis}-region. We rerun the MCMC algorithm on the imputed data set and plot the marginal posterior inclusion probabilities of top-ranked SNPs according to their genomic positions and inferred configurations in Figure \ref{fine_mapping.figure}. 
The plot clearly indicates three adjacent however distinct genomic regions with a much improved resolution. We note that although the individual SNP inclusion probabilities decrease significantly from the previous analysis, the inclusion probabilities of the three regions all increase in some degree: the probability of the LCL only eQTL region increases from 0.66 to 0.68, the probability of the consistent eQTL region increases from 0.89 to 0.95 and the the probability of the Fibroblast only eQTL region increases from 0.38 to 0.61. Figure  \ref{fine_mapping.figure} also shows  SNP genotypes are highly correlated within each region,  and it is impossible to distinguish the true causal variants based on association analysis. Therefore, it seems only logical to report interesting regions rather than individual variants in such settings. 

\begin{figure}[h!t]
\begin{center}
\includegraphics[totalheight=0.60\textheight]{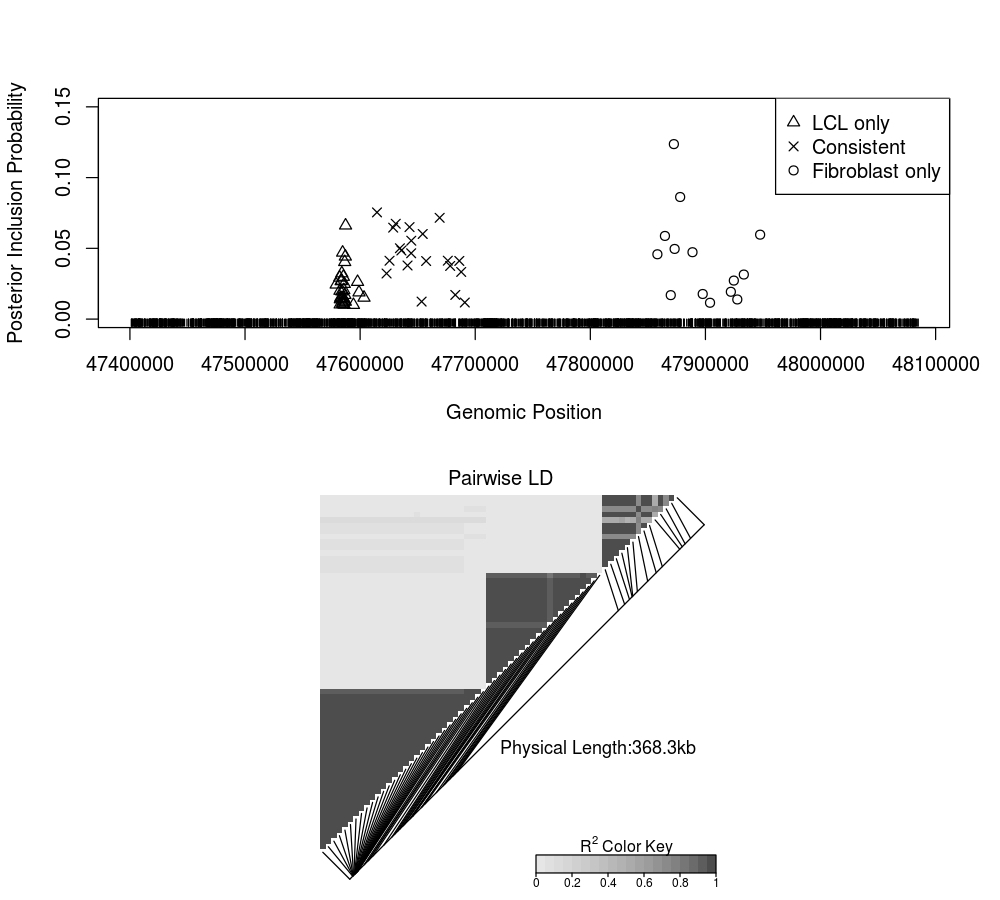}
\caption{\label{fine_mapping.figure} eQTL fine-mapping for gene C21orf57 with a dense SNP set.  The top panel plots SNPs with marginal posterior inclusion probabilities $\ge 0.01$. The different symbols indicate the different activity configurations of potential eQTLs. The ticks on the X-axis label the positions of interrogated SNPs (genotyped and imputed) in this region. Three distinct genomic regions that harbor three different eQTLs with different tissue configurations can be clearly identified from the plot. The inferred high posterior probability models typically contain one SNP from each of the three regions. 
The bottom panel displays the correlations, measured by $r^2$, between the SNPs plotted in the top panel (produced by R package LDheatmap). It should be clear that genotype correlations within each identified genomic region are quite high, and between the regions, the SNPs are much less correlated.}
\end{center}
\end{figure}

\section{Discussion}

The general statistical problem we have considered in this paper is related to the problem of structured variable selection. Our Bayesian approach provides a general framework to specify both group structures (through $\Pr\left(\xi(\bvg)\right)$) and prior correlations on non-zero effects (through $p(\Lambda_g \mid \xi(\bvg))$) in a hierarchical fashion. Compared with regularization-based model selection methods such as group LASSO (\cite{Yuan2006}) and fused LASSO (\cite{Tibshirani2005}), our method is more flexible and conceptually easier to apply. For example, in the multiple-tissue eQTL mapping example, the association patterns within a SNP across multiple tissues are rather complex; neither group LASSO (which encourages the whole group to be selected) nor sparse group LASSO (which encourages only a few members in the group to be selected) is suitable in this context. 

One of our main contributions in this paper is the results involving Bayes factors. Although we have focused mostly on model selection, our novel results can be directly applied to hypothesis-testing settings (e.g., in gene-based genetic association testing). It should be noted that although we have described our results exclusively assuming quantitative Gaussian response variables, our results can be naturally extended to the generalized linear models. We give a brief argument for this extension in appendix G.

Our simulations and data application both focus on the problem of mapping eQTLs across multiple tissues. We note that although many sophisticated statistical methods have been developed for mapping multiple ({\em cis} and {\em trans}) eQTLs (\cite{Scott-Boyer2012, Xu2009}), almost none of them considers the mapping problem in a multiple-tissue context. As shown by \cite{Flutre2012} and \cite{Ding2011}, naively applying single tissue mapping method one tissue at a time not only lacks of power in detecting tissue-consistent eQTLs but also can be dangerous in inferring tissue-specific eQTLs. Our statistical framework naturally fills this gap.  The SSLR model can be naturally applied to the fine-mapping problem in genetic association meta-analysis. Furthermore,  a prior requiring genuine association signals to display low within-group heterogeneity seems most appropriate in this context.

In our examples, we deal with relatively small $r$ value by using a non-informative inverse-Wishart prior. It should be noted that our Bayes factor results can be applied in the situation where $r$ is also high dimensional. In principle, a strong and informative prior on $\Sigv$ is sufficient (see Equation (\ref{bayes.Sigma.est})). However in practice, constructing a reasonable and strongly informative prior for a covariance matrix in high-dimensional is challenging and the context of the applications should always be carefully considered. A useful statistical technique in specifying the inverse-Wishart prior in high dimensional settings is  to utilize its connection to Gaussian graphical models (\cite{Dawid1993, Carvalho2009}) which can be extremely helpful to systematically describe the complex relationships among a large number of variables.

Finally, although we have demonstrated our approach exclusively in the genetic/genomic context, the statistical approaches presented in this paper are general enough to apply to model selection problems in other contexts, such as graphical model inference and Bayesian causal inference, to name a few examples.

\section{Software Distribution}
Software package implementing the computational methods described in this paper is available and actively maintained on the website \url{https://github.com/xqwen/sbams/}. 

\section{Acknowledgment}
We thank Jeremy Taylor, Peter Song, Ji Zhu, Bin Nan, Matthew Stephens, Timothee Flutre and Xiang Zhou, the associate editor and two anonymous referees for valuable comments. This work is supported by NIH grant HG007022 (PI G. Abecasis).

\appendix	

\section{Prior Specification in Genetic Applications}\label{genet.prior.appx}

In this section, we summarize and discuss some of the existing results that can be utilized for the prior specification in the SSMR model in various genetic applications.

\subsection{Prior Decomposition by Genetic Variants}
 
\cite{Guan2011} argue that regression coefficients of genetic effects reflect the ``causal" effects on the phenotype of interest and there is no obvious reason to suspect these causal effects among different variants are correlated spatially. (Note, it is important to distinguish the correlations among the observed genotypes and the independence of the underlying genetic effects.) The similar type of the independent prior has also been widely used in the polygenic models. As a consequence of this reasoning, it is sensible to decompose the $\Wv_g$ matrix into a block diagonal structure, i.e., $\Wv_g = \Phiv_1 \oplus \cdots + \Phiv_p$, where each block matrix $\Phiv_i$ corresponds to a single SNP. Also, the prior distribution $\Pr\left(\xi(\bvg)\right)$ can be factored into the product of the prior probability of each SNP.  

The simple i.i.d priors on SNPs provide a useful starting point for many applications in genetics. More recently,  many authors (\cite{Veyrieras2008, Stingo2011}) have proposed to integrate SNP-level genomic annotation information into prior specifications. In the simplest case,  a logit function is used to connect the genomic feature of a SNP and its marginal prior inclusion probability,  and a ``feature coefficient" is parametrized to quantify the impact of the genomic feature on the genetic association. The feature coefficient in this context is typically unknown and often of great interest for inference. As a consequence, the priors on different SNPs are no longer i.i.d. This approach  not only is useful in integrating additional information to identify the {\em causal} genetic variant, but also provides an elegant parametric framework to perform feature enrichment analysis, i.e., the posterior inference results of the feature coefficients summarize all necessary statistical evidence of the enrichment of association signals in the relevant annotation categories.

\subsection{Priors for Multiple Quantitative Traits Associations}

The interplays between genetic variants and multiple phenotypes are complicated: not only genetic variants can directly affect multiple phenotypes, but also there are  interactions between phenotypes through gene networks. As a result, genetic variants and phenotypes can be interacted in an indirect way (through some intermediate phenotypes).

Most recently, \cite{Stephens2012} proposes a directed acyclic graph (DAG) approach to address the structured phenotype relationships. Their approach first classifies phenotypes into three groups of directly affected, indirectly affected and unaffected with respect to a target genetic variant. Conditioning on the classification, an MVLR model is employed to model the genetic association between the genetic variant and the directly affected phenotypes. 
Because the true classification of the phenotypes is typically unknown, they use Bayesian model averaging technique to account for this latent structure.

Other approaches (\cite{Scott-Boyer2012, Stingo2011}) avoid directly modeling the relationship among multiple phenotypes, instead they utilize prior biological pathway and information of gene networks to prioritize the potential associations of a target variant with respect to a group of phenotypes. 

\subsection{Priors for Heterogeneous Genetic Effects in Subgroups}\label{het.genet.prior.appx}
  
When considering the genetic effects between a genetic variant and a phenotype in various subgroups (formed either by environmental conditions, e.g. in G$\times$E interactions, or by sampling structures, e.g. in meta-analysis), the key is to account for the heterogeneity of genetic effects. \cite{Wen2011} have recently proposed a flexible Bayesian prior to model heterogeneous genetic effects across multiple subgroups. For a genetic variant, this prior assumes that its genetic effects with respect to a common phenotype in $s$ subgroup, if non-zero, are described by 
\begin{equation}
  \beta_i \sim {\rm N}(\bar \beta, \phi^2), i = 1,\dots,s,
\end{equation} 
and
\begin{equation}
  \bar \beta \sim {\rm N}(0, \omega^2),
\end{equation}
where parameter $\omega^2$ quantifies the prior magnitude of the average effect and $\phi^2$ describes the prior degree of heterogeneity.  Equivalently, the joint prior distribution for vector $(\beta_1,\dots,\beta_r)$ can be represented by a multivariate normal distribution with mean 0 and variance-covariance matrix $W_g$, where
 \begin{equation}  
     \label{meta.phi} 
    \Wv_g = \left( \begin{array}{ccc} \phi^2+\omega^2 & \cdots & \omega^2  \\  \vdots & \ddots & \vdots \\ \omega^2 & \cdots &\phi^2+ \omega^2     \end{array}\right).
\end{equation} 
It is easy to see that $\frac{\omega^2}{\omega^2+\phi^2}$ is the prior correlation between a pair of genetic effects: when $\phi^2$ is set to 0, it corresponds to the fixed effect model; whereas setting $\omega^2=0$ implies the effects are {\it a priori} independent in all subgroups.

\section{Bayes Factor Derivation}\label{ssmr.appx}

In this section, we show the derivation of Bayes factors based on the SSMR model.

In the SSMR model, we have defined $\YYv, \XXv, \SSigv, \bvc$ and $\bvg$ in section 2 of the main text. In addition, we denote the complete collection of regression coefficients and its vectorized version by $\BBv := \{\Bv_1,\dots, \Bv_s\}$ and $\bvs :=  {\tiny \left( \begin{array}{c} \bvc \\  \bvg \\ \end{array}\right)}$, respectively.

The likelihood function of the SSMR model is given by
\begin{equation}
   \label{ssmr.likelihood}
   p(\YYv | \XXv, \BBv,\SSigv) = (2 \pi)^{-\frac{r \sum_{i=1}^s n_i}{2}} \cdot \prod_{i=1}^s |\Sigv_i|^{-\frac{n_i}{2}}\cdot \etr \left(-\frac{1}{2}\sum_{i=1}^s \Sigv_i^{-1} (\Yv_i - \Xv_i{\Bv_i})' (\Yv_i - \Xv_i{\Bv_i}) \right)
\end{equation}
where function $\etr(\cdot)$ denotes the exponential of the trace. Given the least squares estimate $\hat B_i$ for each composing MVLR, it follows that
\begin{equation}
    (\Yv_i - \Xv_i{\Bv_i})' (\Yv_i - \Xv_i{\Bv_i})  = (\Yv_i - \Xv_i \hat {\Bv_i})' (\Yv_i -  \Xv_i \hat{\Bv_i}) + (\Bv_i - \hat{\Bv_i})' (\Xv_i'\Xv_i) (\Bv_i - \hat{\Bv_i}).
\end{equation}
Note this decomposition holds even if $\Xv_i$ is rank-deficient (however, $\hat \Bv_i$ may not be unique, see \cite{Mccullagh1989}, page 82 for discussions). We denote  $\bv_i := \vect(\Bv_i')$ and $ \hat \bv_i := \vect(\hat \Bv_i')$, and use  $\bva$ and $\hbva$ to denote the sequentially concatenated vectors of $(\bv_1,\dots,\bv_s)$ and  $(\hat \bv_1,\dots, \hat \bv_s)$, respectively. 
The likelihood function (\ref{ssmr.likelihood}) can be re-written as
\begin{equation}
   \label{ssmr.likelihood2}
    \begin{aligned}
   p(\YYv | \XXv, \BBv,\SSigv) = &(2 \pi)^{-\frac{r \sum_{i=1}^s n_i}{2}} \cdot \prod_{i=1}^s |\Sigv_i|^{-\frac{n_i}{2}} \cdot \etr \left(-\frac{1}{2}\sum_{i=1}^s \Sigv_i^{-1} (\Yv_i - \Xv_i \hat {\Bv_i})' (\Yv_i - \Xv_i \hat {\Bv_i}) \right) \\
                                     & \cdot  \exp \left(-\frac{1}{2} \left(\bva -\hbva \right)' \Phiv \, \left(\bva - \hbva \right) \right), \\
   \end{aligned} 
  \end{equation}
where 
\begin{equation*}
  \Phiv = \left(\Xv_1'\Xv_1\otimes\Sigv_s^{-1}\right)\oplus \cdots \oplus \left(\Xv_s'\Xv_s\otimes\Sigv_s^{-1}\right).
\end{equation*}  
Also, by the general case of Gauss-Markov theorem, we note that ${\rm Var}(\hbva) = \Phiv^{-1}$ (In case that $\Phiv$ is singular, the Moore--Penrose pseudoinverse is applied).  

Although $\bvs$ and $\bva$ generally differ in the orders of the composing elements, they can be reconciled by a permutation operation, i.e.,
\begin{equation}
  \Pv \bva = \bvs,
\end{equation}
where $\Pv$ is a $(rps+r\sum_i^s q_i) \times (rps+r\sum_i^s q_i)$ permutation matrix. 
Furthermore, we denote 
\begin{equation*}
   \Omegav = \Pv \Phiv \Pv,
\end{equation*}
and it can be shown that 
\begin{equation}
   \label{vhbvs}
 {\rm Var}(\hbvs) =  \Omegav^{-1}.
\end{equation}

As a result,
\begin{equation}
  \begin{aligned}
   p(\YYv | \XXv, \bvs,\SSigv) = &(2 \pi)^{-\frac{r \sum_i^s n_i}{2}} \cdot \prod_i^s |\Sigv_i|^{-\frac{n_i}{2}}\cdot \etr \left(-\frac{1}{2}\sum_i^s \Sigv_i^{-1} (\Yv_i - \Xv_i \hat{\Bv_i})' (\Yv_i - \Xv_i \hat{\Bv_i}) \right) \\
                           &~\cdot  \exp \left(-\frac{1}{2} \left(\bvs -\hbvs \right)' \Omegav \left(\bvs - \hbvs \right) \right).
 \end{aligned}
\end{equation}


\subsection{Bayes Factor for Known $\Sigv$}\label{exact.bf.appx}

With $\SSigv$ known, the marginal likelihood $p(\YYv|\XXv,\SSigv)$ can be evaluated analytically, i.e.,
\begin{equation}
  p(\YYv|\XXv,\SSigv) = \int p(\YYv|\XXv,\SSigv,\bvs) p(\bvs)\,d\bvs.
\end{equation}
Recall the prior distribution defined in section 2 of the main text,
\begin{equation*}
  \bvs \sim {\rm N}( {\bf 0}, \Psiv_c \oplus \Wv_g).
\end{equation*}
Assuming $\Wv_g$ is full rank, the integration yields
\begin{equation}
  \label{int.rst1}	
  \begin{aligned}
   	  p(\YYv|\XXv,\SSigv) =  & (2 \pi)^{-\frac{r \sum_i^s n_i}{2}} \cdot \prod_i^s |\Sigv_i|^{-\frac{n_i}{2}} \cdot |\Wv_g|^{-\frac{1}{2}} \cdot |\Psiv_{c}|^{-\frac{1}{2}} \cdot |\Omegav + \Psiv_{c}^{-1} \oplus \Wv_g^{-1}|^{-\frac{1}{2}} \\
   	                    & \cdot \exp \left( - \frac{1}{2} \hbvs' \Omegav \left (  \Omegav^{-1} -   (\Omegav + \Psiv_{c}^{-1} \oplus \Wv_g^{-1})^{-1}\right) \Omegav  \hbvs \right) \\
                                    &\cdot \etr \left(-\frac{1}{2}\sum_i^s \Sigv_i^{-1} (\Yv_i - \Xv_i \hat{\Bv_i})' (\Yv_i - \Xv_i \hat{\Bv_i}) \right) , \\
   \end{aligned}	
\end{equation}	
To further simplify (\ref{int.rst1}), we decompose $\Omega$ into the following block matrix
\begin{equation*}
  \Omegav = \left( \begin{array}{cc}  \Omegav_c &  \Omegav_f \\ \Omegav_f' &  \Omegav_g    \end{array}\right),
\end{equation*}
where $\Omegav_c$ and $\Omegav_g$ match the the dimensions of the matrices $\Psiv_c$ and $\Wv_g$, respectively. By (\ref{vhbvs}), it follows that
\begin{equation}
  \label{vhbvg}
  \Vv_g^{-1} = \Omegav_g - \Omegav_f' \Omegav_c^{-1} \Omegav_f.
\end{equation}
Let 
\begin{equation*}
  \bm{\mathcal{U}} = \Omegav_g - \Omegav_f' ( \Omegav_c + \Psiv_c^{-1})^{-1}\Omegav_f + \Wv_g^{-1},
\end{equation*}
and it follows that
\begin{equation}
  |\Omegav + \Psiv_c^{-1}\oplus \Wv_g^{-1} | = |\Omegav_c+ \Psiv_c^{-1}| \cdot |\bm{\mathcal{U}}|. 
\end{equation}  
Furthermore, the matrix product, $\Omegav \left (  \Omegav^{-1} -   (\Omegav + \Psiv_{c}^{-1} \oplus \Wv_g^{-1})^{-1}\right) \Omegav $, can be represented by the block matrix $\left( \begin{array}{cc} \Av & \Bv \\ \Bv' & \Dv \\ \end{array}\right)$, where
\begin{equation*}
  \begin{aligned}
      & \Av =  \Omegav_c \left[\Iv - (\Omegav_c + \Psiv_c^{-1})^{-1}\Omegav_c \right]  - \left[\Iv - \Omegav_c (\Omegav_c + \Psiv_c^{-1})^{-1}\right]\,\Omegav_f\,\bm{\mathcal{U}}^{-1} \Omegav_f'\,\left[\Iv - (\Omegav_c + \Psiv_c^{-1})^{-1}\Omegav_c \right] ,  \\
       & B = \left[\Iv - \Omegav_c (\Omegav_c + \Psiv_c^{-1})^{-1}\right]\,\Omegav_f\,\bm{\mathcal{U}}^{-1}\Wv_g^{-1}\\
       & D = \Wv_g^{-1} - \Wv_g^{-1} \bm{\mathcal{U}}^{-1} \Wv_g^{-1} = (\bm{\mathcal{U}}-\Wv_g^{-1}) - (\bm{\mathcal{U}}-\Wv_g^{-1})\bm{\mathcal{U}}^{-1}(\bm{\mathcal{U}}-\Wv_g^{-1}).\\
  \end{aligned}
\end{equation*}  
Although the expressions are fairly complicated, when the limit $\Psiv_c^{-1} \to 0$ is taken, $\Av \to 0$ and $\Bv \to 0$.

The exact same calculations can be carried out with respect to the null model. In the end, we obtain the following marginal likelihood under $H_0$,
\begin{equation}
  \label{int.rst0}	
  \begin{aligned}
   	  P(\YYv|\XXv,\SSigv,H_0) & =  (2 \pi)^{-\frac{r \sum_i^s n_i}{2}} \cdot \prod_i^s |\Sigv_i|^{-\frac{n_i}{2}} \cdot |\Psiv_c|^{-\frac{1}{2}} \cdot |\Omegav_c + \Psiv_{c}^{-1}|^{-\frac{1}{2}} \\
                           & \cdot \exp \left( - \frac{1}{2} \tbvc' \Omegav_c \left ( \Omegav_c^{-1} -  (\Omegav_c+ \Psiv_c^{-1})^{-1}\right) \Omegav_c  \tbvc \right) \\
                          & \cdot \etr \left(-\frac{1}{2} \sum_i^s \Sigv_i^{-1} (\Yv_i - \Xv_{c,i} \tilde \Bv_i)'(\Yv_i - \Xv_{c,i} \tilde \Bv_i)\right), \\
  \end{aligned}
\end{equation}   
where $\tbvc$  and $\tilde \Bv_i$ are the MLEs of regression coefficients obtained under the null model (i.e. restricting $\bvg \equiv 0$). Note the relationship of the least squares estimates between the target and the null models:
\begin{equation}
  \tilde \Bv_i = \hat \Bv_{c,i} + (\Xv_{c,i}'\Xv_{c,i})^{-1} \Xv_{c,i}' \Xv_{g,i} \hat \Bv_{g,i},
\end{equation}
and
\begin{equation}
  \label{residual.relation}
  \begin{aligned}
  &~~~(\Yv_i - \Xv_{c,i} \tilde \Bv_i)'(\Yv_i- \Xv_{c,i} \tilde \Bv_i) - (\Yv_i -\Xv_i \hat \Bv_i)'(\Yv_i- \Xv_i \hat \Bv_i) \\
    & = {\hat \Bv_{g,i}}'\left(\Xv_{g,i}'\Xv_{g,i}-\Xv_{g,i}'\Xv_{c,i}(\Xv_{c,i}'\Xv_{c,i})^{-1}\Xv_{c,i}'\Xv_{g,i}\right) \hat \Bv_{g,i} \\
  \end{aligned} 
\end{equation}
It follows that 
 \begin{equation}
   \label{relation.eqn}
   \begin{aligned}
   &~~~\etr \left(\frac{1}{2} \sum_i^s \Sigv_i^{-1} \left[(\Yv_i - \Xv_{c,i} \tilde \Bv_i)'(\Yv_i- \Xv_{c,i} \tilde \Bv_i) - (\Yv_i - \Xv_i \hat \Bv_i)'(\Yv_i- \Xv_i \hat \Bv_i)\right]  \right) \\
   & = \exp\left(\frac{1}{2} \hbvg' \Vv_g^{-1} \hbvg \right).
   \end{aligned}
 \end{equation}
This also gives the explicit expression for $\Vv_g^{-1}$, i.e.,
\begin{equation}
  \Vv_g^{-1} = \oplus_{i=1}^s \Vv_{g,i}^{-1} = \oplus_{i=1}^s \left[\left(\Xv_{g,i}'\Xv_{g,i}-\Xv_{g,i}'\Xv_{c,i}(\Xv_{c,i}'\Xv_{c,i})^{-1}\Xv_{c,i}'\Xv_{g,i}\right) \otimes \Sigv_i^{-1}\right].
\end{equation}
Because of the block-diagonal nature of the $\Vv_g^{-1}$ matrix, the following expression also holds true
\begin{equation}
   \label{relation.eqn2}
   \hbvg' \Vv_g^{-1} \hbvg = \sum_{i=1}^s \hat \bv_{g,i}'  \Vv_{g,i}^{-1} \hat \bv_{g,i},
\end{equation}   
which provides convenience for Laplace approximation later on.

Finally, by taking the limit $\Psiv_c^{-1} \to 0$ and noting 
\begin{equation}
  \lim_{\psi_c^{-1} \to 0} \bm{\mathcal{U}} = \Vv_g^{-1} + \Wv_g^{-1},
\end{equation}
we obtain 
\begin{equation}
   \BF(\Wv_g) = |\Iv +  \Vv_g^{-1}\Wv_g|^{-\frac{1}{2}} \cdot \exp \left( \frac{1}{2} \hbvg' \Vv_g^{-1}\left[\Wv_g(\Iv +  \Vv_g^{-1}\Wv_g)^{-1}\right] \Vv_g^{-1} \hbvg \right),
\end{equation}
which proves LEMMA 1.

\subsection{Approximate Bayes Factors for Unknown $\Sigv$}\label{lap.appx}

When $\SSigv$ is unknown, we assign independent inverse Wishart priors, ${\rm IW}_r(\nu_i \Hv_i,m_i)$, to each $\Sigv_i$ and additional integrals are required for computing the marginal likelihood. More specifically, the goal is to evaluate
\begin{equation}
   p(\YYv|\XXv) = \int p(\YYv|\XXv, \SSigv) \prod_i p(\Sigv_i^{-1}) ~d\, \Sigv_1^{-1}\dots d\, \Sigv_s^{-1},
\end{equation}
where
\begin{equation}
   p(\Sigv_i^{-1}) \propto |\Sigv_i^{-1}|^{\frac{m_i-r-1}{2}} \etr\left(-\frac{1}{2} \nu_i \Hv_i\Sigv_i^{-1}\right).
\end{equation}    

The desired Bayes factor is therefore computed as
\begin{equation}
  \BF(\Wv_g) = \lim_{\Psiv_c^{-1} \to 0} \frac{\int p(\YYv|\XXv,\SSigv)\prod_i p(\Sigv_i^{-1}) ~d\, \Sigv_1^{-1}\cdots\, d\, \Sigv_s^{-1} }{\int p(\YYv|\XXv,\SSigv, H_0)\prod_i p(\Sigv_i^{-1}) ~d\, \Sigv_1^{-1}\cdots\, d\, \Sigv_s^{-1}}.
\end{equation}
By plugging in (\ref{int.rst1}) and (\ref{int.rst0}) and noting the cancellation of $|\Psiv_c|^{-\frac{1}{2}}$ terms along with the fact that 
$\Omegav^{-1} -   (\Omegav + \Psiv_{c}^{-1} \oplus \Wv_g^{-1})^{-1}$
is positive definite, it is easy to see that the remaining integrands, both are functions of $\Psiv_c^{-1}$, are bounded. It is then justified by bounded convergence theorem (BCT) to switch the limit and integration operations. As a result, we obtain 
\begin{equation}
   {\rm BF}(\Wv_g) = \frac{\int K_{H_a} \,d\, \Sigv_1^{-1}\cdots \, d\,\Sigv_s^{-1}}{\int K_{H_0} \,d\, \Sigv_1^{-1}\cdots \,d\, \Sigv_s^{-1}},
\end{equation}
where
\begin{equation}
  \label{mvlr.kha}
   \begin{aligned} 
   K_{H_a}   &= |\Iv + \Vv_g^{-1}\Wv_g|^{-\frac{1}{2}} \cdot \exp \left(\frac{1}{2} \hbvg'  \left[\Vv_g^{-1}\Wv_g(\Iv+\Vv_g^{-1}\Wv_g)^{-1} \Vv_g^{-1}\right] \hbvg \right) \\
                     & \cdot \prod_{i=1}^s |\Sigv_i^{-1}|^{\frac{n_i+m_i-q_i-r-1}{2}} \cdot \etr \left(-\frac{1}{2} \sum_{i=1}^s \Sigv_i^{-1} \left(\nu_i \Hv_i+(\Yv_i - \Xv_{c,i} \tilde \Bv_i )'(\Yv_i - \Xv_{c,i} \tilde \Bv_i)\right)\right),\\
   \end{aligned} 
\end{equation}    	 
\begin{equation} 
   \label{mvlr.kh0}
   K_{H_0}  = \prod_{i=1}^s |\Sigv_i^{-1}|^{\frac{n_i+m_i-q_i-r-1}{2}} \cdot \etr \left(-\frac{1}{2} \sum_{i=1}^s \Sigv_i^{-1} \left(\nu_i \Hv_i+(\Yv_i - \Xv_{c,i} \tilde \Bv_i )'(\Yv_i - \Xv_{c,i} \tilde \Bv_i)\right)\right), 
\end{equation}

Because $\Vv_g^{-1}$ and (potentially) $\Wv_g$ are both functions of $\SSigv$, the analytic integration of $K_{H_a}$ is generally implausible. Here we approximate the integrals of both $K_{H_a}$ and ${K_{H_0}}$ by Laplace's method. Note, although the analytic integration of $K_{H_0}$ is straightforward, it is been shown (\cite{Wen2011}) that simultaneously applying Laplace's methods to both $K_{H_a}$ and ${K_{H_0}}$ achieves better numerical accuracy for desired Bayes factor. 

Laplace's method approximates an integral with respect to a $ d \times d$ symmetric matrix $\Zv$ (or equivalently the corresponding half-vectorized $(d+1)d/2$ dimensional vector $\rm{vech}(\Zv)$) in the following way,   
\begin{equation}
 \label{laplace.approx} 
  \int_D h(\Zv) \exp \left(\, g(\Zv) \,\right)\, d\,\Zv \approx (2 \pi)^{d(d+1)/4}|\Hv_{\tiny \hat \Zv}|^{-1/2}h(\hat \Zv)\exp \left(\,g(\hat \Zv)\,\right),
\end{equation} 
where 
\begin{equation*}
 \hat \Zv = \arg\max_{\small \Zv} g(\Zv),
\end{equation*}
and $| \Hv_{\tiny \hat \Zv}|$ is the absolute value of the determinant of the Hessian matrix of the function $g$ evaluated at $\hat \Zv$.  The technical requirements on the factorization  are that $h(\cdot)$  is smooth and positively valued and $g(\cdot)$ is  smooth and obtains its unique maximum in the interior of $D$. Although different factorization schemes generally achieve different approximation accuracies for finite sample sizes, the asymptotic error bounds are typically the same. For a detailed discussion, see \cite{Butler2007} chapter 2.
 
 To evaluate the desired Bayes factor, we sequentially apply the Laplace's method with respect to each $\Sigv_i^{-1}$ for both $K_{H_a}$ and $K_{H_0}$.

\subsubsection{General Derivation}

By (\ref{relation.eqn}) and (\ref{relation.eqn2}), we note the exponential term
\begin{equation}
  {\rm tr}\left[\Sigv_j^{-1} \left(\nu_j \Hv_j + (\Yv_j - \Xv_{c,j} \tilde \Bv_j)'(\Yv_j- \Xv_{c,j} \tilde \Bv_j)\right)\right],
\end{equation}
is presented in both the alternative and the null models for each multivariate linear regression model $j$, and it can be generally decomposed into 
\begin{equation}
 \begin{aligned}
  &{\rm tr} \left[ \Sigv_j^{-1}\left(\nu_j \Hv_j + (1-\alpha_j)(\Yv_j - \Xv_{c,j} \tilde \Bv_j)'(\Yv_j- \Xv_{c,j} \tilde \Bv_j)+ \alpha_j\,(\Yv_j - \Xv_{c,j} \hat \Bv_j)'(\Yv_j- \Xv_{c,j} \hat \Bv_j)\right)\right] \\
   &+ \alpha_j\, \hat \bv_{g,i}' \Vv_{g,i}^{-1} \hat \bv_{g,i},\\
 \end{aligned}
\end{equation}
where $\alpha_j \in [0,1]$. 
Thus, when applying Laplace's method, we start by factoring $K_{H_a}$ into
\begin{equation}
  \label{lap.fac}
 K_{H_a} =  h_a(\Sigv_1^{-1},\dots,\Sigv_s^{-1})  \exp \left(\,g_a(\Sigv_1^{-1},\dots,\Sigv_s^{-1})\,\right),
\end{equation}
where
\begin{equation}
   \begin{aligned}
    h_a(\Sigv_1^{-1},\dots,\Sigv_s^{-1}) &=  |\Iv + \Vv_g^{-1}\Wv_g|^{-\frac{1}{2}} \cdot \exp \left(\frac{1}{2} \hbvg' \left[\Vv_g^{-1}\Wv_g(\Iv+\Vv_g^{-1}\Wv_g)^{-1} \Vv_g^{-1}\right]\hbvg \right)  \\
                      & \cdot \exp\left(- \frac{1}{2} \sum_{i=1}^s \alpha_i \hat \bv_{g,i}' \Vv_{g,i}^{-1} \hat \bv_{g,i}\right)
   \end{aligned} 
\end{equation}
and
\begin{equation}
  \begin{aligned}
  &~~~~g_a (\Sigv_1^{-1},\dots,\Sigv_s^{-1})  = \sum_{i=1}^s \frac{n_i+\nu_i}{2} \log |\Sigv_i^{-1}|   \\
                         &-\frac{1}{2}\sum_{i=1}^s {\rm tr} \left[ \Sigv_i^{-1} \left(\nu_i \Hv_i+\alpha_i (\Yv_i - \Xv_i \hat \Bv_i)'(\Yv_i- \Xv_i \hat \Bv_i)+ (1-\alpha_i)(\Yv_i - \Xv_i \tilde \Bv_i)'(\Yv_i- \Xv_i \tilde \Bv_i)\right)\right].
  \end{aligned}
\end{equation}

It is straightforward to show that the unique maximum of $g(\Sigv^{-1},\dots,\Sigv^{-1}_s)$ can be obtained by performing sequential analytic maximization with respect to each individual $\Sigv_i$. More specifically, the maximum is attained at
\begin{equation}
  \label{bayes.Sigma.est.appx}
  \check \Sigv_i = \frac{\nu_i}{n_i + \nu_i} \Hv_i + \frac{n_i}{n_i + \nu_i} \left[\alpha_i \hat \Sigv_i + (1-\alpha_i) \tilde \Sigv_i \right],~\forall i,
\end{equation}
where 
\begin{equation}
  \hat \Sigv_i =  \frac{1}{n_i}(\Yv_i - \Xv_i \hat \Bv_i)'(\Yv_i- \Xv_i \hat \Bv_i),
\end{equation}
and
\begin{equation}
  \tilde \Sigv_i =  \frac{1}{n_i}(\Yv_i - \Xv_{c,i} \tilde \Bv_i)'(\Yv_i- \Xv_{c,i} \tilde \Bv_i),
\end{equation}
are commonly used MLEs of $\Sigv_i$ evaluated under the target and the null models respectively.

 Following \cite {Minka2000}, it can be shown that the Hessian matrix $H_{g_a}(\Sigv_i^{-1})$ for each $\Sigv_i^{-1}$ is given by
\begin{equation}
  \begin{aligned}
  H_{g_a}(\Sigv_i^{-1}) =& \frac {{\rm d}^2 g_a}{{\rm dvech}(\Sigv_i^{-1}) {\rm dvech}(\Sigv_i^{-1})'}\\
                       =& -\frac{n_i}{2} \Dv_s'\left(\Sigv_i \otimes \Sigv_i \right) \Dv_s,
  \end{aligned}
\end{equation}       
where $\Dv_s$ denotes the duplication matrix for $s \times s$ symmetric matrices. As it is evaluated at $\check \Sigv_i^{-1}$, its absolute determinant results in the following simple form,
\begin{equation}
    |H_{g_a}(\check \Sigv_i^{-1})| = 2^{-r} n_i^{r(r+1)/2} |\check \Sigv_i|^{r+1}.
\end{equation}    

Similarly, we factor $K_{H_0}$ in the same way, i.e.,
\begin{equation}
 K_{H_0} =  h_0(\Sigv_1^{-1},\dots,\Sigv_s^{-1})  \exp (\,g_0(\Sigv_1^{-1}),\dots,\Sigv_s^{-1}\,),
\end{equation}
where 
\begin{equation}
   \label{kh0.h0}
   \begin{aligned}
   h_0(\Sigv_1^{-1},\dots,\Sigv_s^{-1}) = \exp\left(- \frac{1}{2} \sum_{i=1}^s \alpha_i \hat \bv_{g,i}' \Vv_{g,i}^{-1} \hat \bv_{g,i}\right)
  \end{aligned}
\end{equation}
and
\begin{equation}
   \label{kh0.g0}
     \begin{aligned}
  &~~~~g_0 (\Sigv_1^{-1},\dots,\Sigv_s^{-1})  = \sum_{i=1}^s \frac{n_i+\nu_i}{2} \log |\Sigv_i^{-1}|   \\
                         &-\frac{1}{2}\sum_{i=1}^s {\rm tr} \left[ \Sigv_i^{-1} \left(\nu_i \Hv_i+\alpha_i (\Yv_i - \Xv_i \hat \Bv_i)'(\Yv_i- \Xv_i \hat \Bv_i)+ (1-\alpha_i)(\Yv_i - \Xv_i \tilde \Bv_i)'(\Yv_i- \Xv_i \tilde \Bv_i)\right)\right].
  \end{aligned}
\end{equation}   	
Note that $g_0(\Sigv_1^{-1},\dots,\Sigv_s^{-1})$ and $g_a(\Sigv_1^{-1}\dots,\Sigv_s^{-1})$ are identical, $(\check \Sigv_1, \dots, \check \Sigv_s)$ also uniquely maximizes $g_0$ function.

Following (\ref{laplace.approx}), the desired Bayes factor is computed as 
 \begin{equation}
   \label{approx.bf1}
     \begin{aligned}
   \BF(\Wv_g)  =  &  |\Iv + \check \Vv_g^{-1} \check \Wv_g|^{-\frac{1}{2}}  \cdot \exp \left(\frac{1}{2} \hbvg' \check \Vv_g^{-1} \left[\check \Wv_g(\Iv+\check \Vv_g^{-1}\check \Wv_g)^{-1} \right] \check \Vv_g^{-1} \hbvg \right) \cdot \prod_{i=1}^s \left(1+O(\frac{1}{n_i})\right)
    \end{aligned}
\end{equation} 
where $\check \Vv_g^{-1}$ and $\check \Wv_g$ are the corresponding $\Vv_g^{-1}$ and $\Wv_g$ evaluated at $(\check \Sigv_1,\dots,\check \Sigv_s)$. In particular,
\begin{equation}
  \check \Vv_g^{-1} = \oplus_{i=1}^s \left[\left(\Xv_{g,i}'\Xv_{g,i}-\Xv_{g,i}'\Xv_{c,i}(\Xv_{c,i}'\Xv_{c,i})^{-1}\Xv_{c,i}'\Xv_{g,i}\right) \otimes \check \Sigv_i^{-1}\right].
\end{equation}
This leads to the final expression of $\ABF$
\begin{equation}
      \label{abf1.rst}
       \ABF(\Wv_g, \av) = |\Iv+ \check \Vv_g^{-1}\check \Wv_g|^{-\frac{1}{2}}\cdot \exp \left( \frac{1}{2} \hbvg' \check \Vv_g^{-1} \left[\check \Wv_g( \Iv + \check \Vv_g^{-1} \check \Wv_g)^{-1}\right] \check \Vv_g^{-1} \hbvg \right),
\end{equation}  
which also completes the proof for PROPOSITION 1.

\section{Computational Stability of Bayes Factor}\label{computing.appx}

In this section, we demonstrate the computational stability of the derived Bayes factors. In particular,  we show that the derived Bayes factor and its approximations can be stably evaluated even if some design matrix  $\Xv_i \in \XXv$ is rank deficient.

First, assuming $\Xv_{c,i}'\Xv_{c,i}$ can be inverted in the general sense  $\forall i = 1,\dots,s$, we define
\begin{equation}
  \Gv_i = \left(\Iv - \Xv_{c,i} \left(\Xv_{c,i}'\Xv_{c,i}\right)^{-1}\Xv_{c,i}'\right)\Xv_{g,i},
\end{equation}
and denote its $p \times n_i$ Moore-Penrose pseudo inverse matrix by $\Gv_i^{+}$. By the general least squares theory, it can be shown (regardless if $\Gv_i$ is full-rank) that  
\begin{align}
   &\hat \Bv_{g,i} = \Gv_i^{+} \Yv_i, \\
   &\hat \bv_{g,i} = \vect(\hat \Bv_{g,i}') = (\Gv_i^{+} \otimes \Iv) \vect(\Yv_i')\\
\end{align}
and 
\begin{equation}
  \Vv_{g,i}^{-1} = \left(\Gv_i'\Gv_i\right)\otimes \Sigv_i^{-1}. 
\end{equation}  
It is then follows from the general property of Moore-Penrose pseudo inverse, such that
\begin{equation}
  \begin{aligned}
   \Vv_{g,i}^{-1} \hat \bv_{g,i} & = \left[\left(\Gv_i'\Gv_i\Gv_i^{+}\right)\otimes \Sigv_i^{-1} \right] \vect(\Yv_i') \\
                                &= (\Gv_i'\otimes\Sigv_i^{-1})\vect(\Yv_i')  \\                                
                                & = \vect(\Sigv^{-1} \Yv_i' \Gv_i). 
  \end{aligned} 
\end{equation}

Finally, $\Vv_g^{-1} \hbvg$ is computed by sequentially concatenating $\Vv_{g,i}^{-1} \hat \bv_{g,i}$ for $i=1,...,s$. 
Note,  in this computational procedure
\begin{enumerate}
  \item there is no matrix inversion operation on $\Xv_{g,i}'\Xv_{g,i}$ (which we allow to be rank deficient).
  \item there is no matrix inversion operation on $\Wv_g$.
  \item matrix $(\Iv + \Vv_g^{-1}\Wv_g)$ is guaranteed positive definite.
\end{enumerate}  

In case that $\SSigv$ is unknown and some $\Xv_{g,i}$ is rank deficient, it becomes inevitable to perform Moore-Penrose pseudo-inverse of $\Xv_i'\Xv_i$ for evaluation of $\hat \Sigv_i$. This would cost the computational efficiency but unlikely affect the computational stability of the $\ABF$.

\section{Computing Bayes Factors with Singular $\Wv_g$}\label{singular.prior.appx}

We first give the proof for PROPOSITION 2 in below.
\begin{proof}  
 In case that $\SSigv$ is known, the proof is trivial by noting that there is no matrix inversion of $\Wv_g$ in the Bayes factor formula of LEMMA 1. 
 
If $\SSigv$ is unknown, the desired Bayes factor is computed by
  \begin{equation}
    \BF(\Wv_g) = \frac{\lim_{\lambda \to 0} \int K_{H_a}(\Wv_g^\dagger(\lambda)) \,d \Sigv_1^{-1}\,...\,d \Sigv_s^{-1}} {\int K_{H_0} \,d \Sigv^{-1}...\,d \Sigv_s^{-1}}, 
  \end{equation}
where the integrands $K_{H_a}$ and $K_{H_0}$ are defined in (\ref{mvlr.kha}) and (\ref{mvlr.kh0}) respectively. It should be clear that 
\begin{equation}
  K_{H_a}(\Wv_g^\dagger(\lambda)) \le \prod_{i=1}^{s} \left[|\Sigv_i^{-1}|^{\frac{n_i+m_i-q_i-r-1}{2}} \cdot \etr \left(-\frac{1}{2} \Sigv_i^{-1} \left( \Hv_i + (\Yv_i - \Xv_i \hat \Bv_i)'(\Yv_i- \Xv_i \hat \Bv_i)\right)\right)\right].
\end{equation}
Because the RHS is clearly integrable with respect to $\Sigv_1^{-1},\dots,\Sigv_s^{-1}$, by bounded convergence theorem, it follows that 
 \begin{equation}
    \BF(\Wv_g) = \frac{ \int \lim_{\lambda \to 0} K_{H_a}(\Wv_g^\dagger(\lambda)) \,d \Sigv_1^{-1}\,...\,d \Sigv_s^{-1}}{\int K_{H_0} \,d \Sigv_1^{-1}\,...\,d \Sigv_s^{-1}}. \\
  \end{equation}
Because the computation of $K_{H_a}$ does not require inversion of $\Wv_g$ and the matrix sum $(\Iv + \Vv_g^{-1}\Wv_g)$ is guaranteed to be full rank, we conclude that
\begin{equation}
   \lim_{\lambda \to 0} K_{H_a}(\Wv_g^\dagger(\lambda))  = K_{H_a}(\Wv_g),
\end{equation}
and
\begin{equation}
  \label{bf.general.singular.wg}
   \lim_{\lambda \to 0}\BF(\Wv_g^\dagger(\lambda)) = \BF(\Wv_g) = \frac{ \int K_{H_a}(\Wv_g)\,d \Sigv_1^{-1}...\,d \Sigv_s^{-1} }{\int K_{H_0} \,d \Sigv_1^{-1}...\,d \Sigv_s^{-1}}, \\
  \end{equation}
provided that $\Wv_g$ is positive semidefinite.
\end{proof}

In case $\Wv_g$ is singular, to evaluate the approximate Bayes factor using Laplace's method, we modify the factorization in (\ref{lap.fac}) to account for the imposed linear restrictions. More specifically, we factor $K_{H_a}$ into 
\begin{equation}
   \begin{aligned}
    h_a(\Sigv_1^{-1},\dots,\Sigv_s^{-1}) &=  |\Iv + \Vv_g^{-1}\Wv_g|^{-\frac{1}{2}} \cdot \exp \left(\frac{1}{2} \hbvg' \left[\Vv_g^{-1}\Wv_g(\Iv+\Vv_g^{-1}\Wv_g)^{-1} \Vv_g^{-1}\right]\hbvg \right)  \\
                      & \cdot \exp\left(- \frac{1}{2} \sum_{i=1}^s \alpha_i \hat \bv_{g,i}^{r'} \Vv_{g,i}^{-1} \hat \bv_{g,i}^r \right),
   \end{aligned} 
\end{equation}
and
\begin{equation}
  \begin{aligned}
  &~~~~g_a (\Sigv_1^{-1},\dots,\Sigv_s^{-1})  = \sum_{i=1}^s \frac{n_i+\nu_i}{2} \log |\Sigv_i^{-1}|   \\
                         &-\frac{1}{2}\sum_{i=1}^s {\rm tr} \left[ \Sigv_i^{-1} \left(\nu_i \Hv_i+\alpha_i (\Yv_i - \Xv_i \hat \Bv_i^r)'(\Yv_i- \Xv_i \hat \Bv_i^r)+ (1-\alpha_i)(\Yv_i - \Xv_i \tilde \Bv_i)'(\Yv_i- \Xv_i \tilde \Bv_i)\right)\right],
  \end{aligned}
\end{equation}  
where $\hat \Bv^r_i$ is the least squares estimate of $\Bv_i$ subject to the linear constraints imposed by $\Wv_g$ and $\hat \bv_{g,i}^r$ is the corresponding vectorized estimates. The remaining steps for Laplace's method are the same as we have shown in appendix B.2.1, however $\hat \Sigv_i$ is now taking the following form: 
\begin{equation}
   \hat \Sigv_i = \frac{1}{n_i}\left(\Hv_i+(\Yv_i - \Xv_i \hat \Bv_i^r)'(\Yv_i- \Xv_i \hat \Bv^r_i)\right).
\end{equation}

\section{Numerical Evaluation of Approximate Bayes Factors}

We perform numerical experiments to assess the finite-sample accuracies of the derived approximate Bayes factors. 

We simulate data under the SSLR model (mainly because its Bayes factors can be numerically evaluated using the adaptive Gaussian quadrature method as the number of groups ($s$) is small). Except for the very last case, our simulated data sets always have sample size $n=75$ and subgroup number $s=3$. We also vary the number of covariates for $p=2,4,8$ and $16$ in different simulations.

For each  $(n,p)$ combination, we simulate 500 data sets using the SSLR model. We intentionally choose small to modest effect sizes, for which accuracies of the Bayes factors matter most. For every simulated data set, we evaluate its ``true value" using the adaptive Gaussian quadrature procedure implemented in the GNU Scientific Library (GSL) and compare it with the $\ABF$s computed under $\av = 0, 0.5$ and $1$. These results are summarized in Figure \ref{abf1.fig} and Table \ref{abf.tbl}. 
As values of $\av$ are set to $0.5$ for all subgroups, the resulting $\ABF$s yield most accurate approximations in all cases with small sample sizes. In comparison, setting $\av=1$ tends to yield anti-conservative approximations whereas setting $\av=0$ leads to conservative approximations.    

\begin{figure}[h!t]
\begin{center}
\includegraphics[totalheight=0.60\textheight]{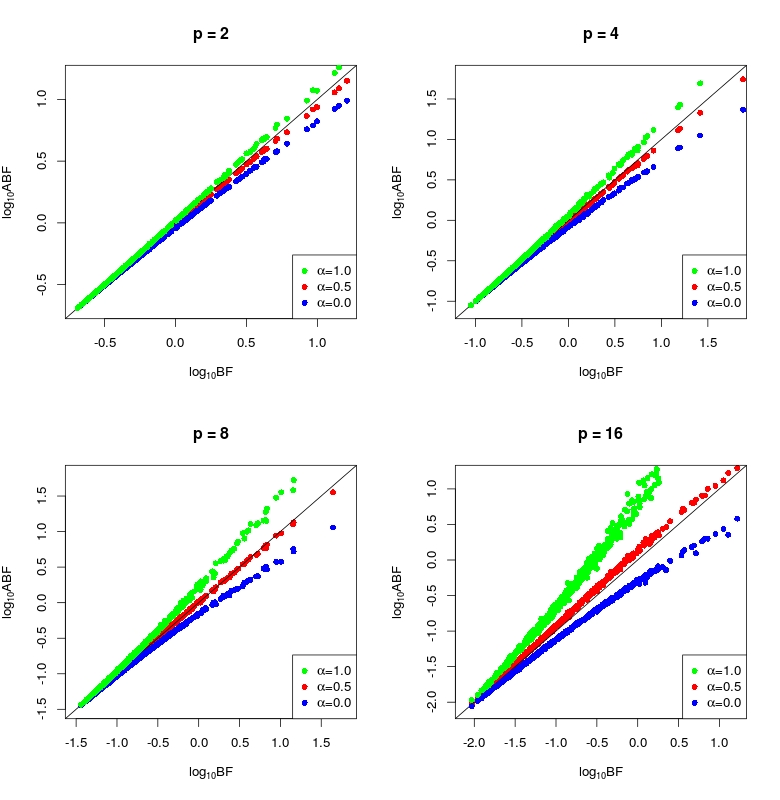}
\caption{\label{abf1.fig} Accuracy of the approximate Bayes factors with small sample sizes. Each data point on the plots represents a single comparison of the $\ABF$ of certain $\av$ value with the true value using a data set simulated from the SSLR model ($n=75$ and $s=3$). The four different panels represent the different numbers of covariates ($p$) allowed in the model.}
 \end{center}
\end{figure}

\begin{table}[ht]
\begin{center}
\begin{tabular} { | l | c c c |  }
\hline
~ & \multicolumn{3}{c |}{RMSE of $\log_{10}(\ABF)$}\\
\cline{2-4}

~ & $\av =0$ & $\av=0.5$ &  $\av=1.0$ \\
\hline 
$n=75,p=2$  &  0.032 & 0.009 & 0.016 \\
$n=75,p=4$  &  0.052 & 0.011 & 0.041 \\
$n=75,p=8$  &  0.074 & 0.008 & 0.096 \\
$n=75,p=16$ &  0.102 & 0.035 & 0.268 \\
\hline
$n=1000, p=16$ & 0.044 & 0.006 & 0.032\\

\hline  

\end{tabular}
\caption{\label{abf.tbl}Root Mean Square Errors (RMSE) of $\log_{10}(\ABF)$ for different $\av$ values under different model settings. The approximate Bayes factors are computed based on the SSLR model with three subgroups ($s = 3$) and different $(n,p)$ settings. Under each setting, we compute $\log_{10}(\ABF)$ for $\av = 0, 0.5, 1.0$ and report the RMSE by comparing the approximations with the true values.}
\end{center}
\end{table}

Finally, to demonstrate a situation that is close to the preferred asymptotic settings, we simulate data for $n=1000$ and $p=16$. The result is shown in Figure \ref{abf2.fig}. It suggests as the sample size increases, all approximations become quite accurate. 

\begin{figure}[h!t]
\begin{center}
\includegraphics[totalheight=0.4\textheight]{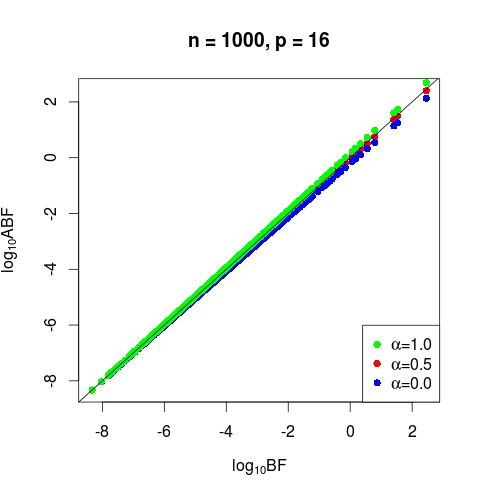}
\caption{\label{abf2.fig} Accuracy of the approximate Bayes factors when the sample size is relatively large. In this plot, the data simulated from the SSLR model with $n=1000, s=3$ and $p=16$. Approximate Bayes factors computed using  different $\av$ values all show good agreement with the true values.}
 \end{center}
\end{figure}

\section{Bayes Factor, Multivariate Test Statistics and the BIC}

In this section, we show that the derived Bayes factor and its approximations are connected to various frequentist multivariate test statistics and the BIC.                                                                                                                    

\subsection{Connection to Multivariate Test Statistics}
Under the following prior specification for $\Wv_g$:
\begin{enumerate}
  \item $\Wv_g =  c \Vv_g$, where $c$ is a positive scalar constant.
  \item $\Vv_g$ is full-rank
  \item $\Sigv_i \sim {\rm IW}(\nu_i\Hv_i, m_i)$ under the limiting conditions, $\nu_i \to 0, \forall i = 1,\dots,s$
\end{enumerate} 
It can be shown that
\begin{equation}
    \label{abf.pp}
    \ABF(\Wv_g, \av=1) = \left(\sqrt{\frac{1}{1+c}}\,\right)^{rps} \cdot \exp\left(\frac{1}{2} \cdot \frac{c}{1+c} \cdot T_{\rm wald} \right),
\end{equation}
and
\begin{equation}
    \label{abfs.pp}
    \ABF(\Wv_g,\av=0) = \left(\sqrt{\frac{1}{1+c}}\,\right)^{rps} \cdot \exp\left(\frac{1}{2} \cdot \frac{c}{1+c} \cdot T_{\rm score} \right),
\end{equation}
where $T_{\rm wald}$ and $T_{\rm score}$ represent the multivariate Wald statistic and the Rao's score statistic, respectively. Both statistics can be used for testing $H_0: \bvg = 0$ based on the SSMR model. Obtaining (\ref{abf.pp}) is straightforward. To establish (\ref{abfs.pp}), we compute the score statistic following \cite{Chen1983}. This yields  
\begin{equation}
  \begin{aligned}
  T_{\rm score} &= \sum_{i=1}^s \vect[(\Yv_i-\Xv_{c,i}\tilde \Bv_i)']'\left(\Xv_{g,i}'\Xv_{g,i} \otimes \tilde \Sigv_i^{-1} \right)\vect[(\Yv_i-\Xv_{c,i}\tilde \Bv_i)']\\
                & = \tbvc' \left[ \oplus_{i=1}^s\left(\Xv_{g,i}'\Xv_{g,i} \otimes \tilde \Sigv_i^{-1} \right)\right] \tbvc,
   \end{aligned}            
\end{equation}  
where $\tilde \Bv_i$ and $\tbvc$ are MLEs of $\Bv_i$ and $\bvc$ estimated from the null model, respectively. Under the specified conditions,
\begin{equation*}
  \begin{aligned}
   &\hbvg' \tilde \Vv_g^{-1} \left[ \tilde \Wv_g (\Iv+ \tilde \Vv_g^{-1} \tilde \Wv_g)^{-1} \right] \tilde \Vv_g^{-1} \hbvg \\
   =& \frac{c}{1+c} \cdot  \hbvg' \tilde \Vv_g^{-1} \hbvg \\
    =& \frac{c}{1+c} \cdot \tbvc' \left[ \oplus_{i=1}^s\left(\Xv_{g,i}'\Xv_{g,i} \otimes \tilde \Sigv_i^{-1} \right)\right] \tbvc.
  \end{aligned}
\end{equation*}

As a consequence, the approximate Bayes factors and the corresponding frequentist test statistics yield the same ranking for a set of candidate models. 

Albeit the connections, we do not advocate the use of these test statistics as model comparison devices in practice. Especially, caution should be taken when interpreting this prior in specific contexts: for example, \cite{Wakefield2009} and \cite{Wen2011} have shown some counter-intuitive implications of this prior in genetic applications (e.g., $|\Wv_g|$ is inversely proportional to sample sizes).

\subsection{Connections to the BIC} \label{bic.appx}

Under the conditions that  
\begin{enumerate}
 \item $\Vv_g$ and $\Wv_g$ are full-rank.
 \item $\lim_{n_i \to 0} \frac{\log |\Wv_g|}{n_i} = 0, \forall i$. 
 \item $n_i \gg p, r, s, \forall i$,.
\end{enumerate}
We show that the BIC can be derived as a rough approximation to the Bayes factor and its approximations under the SSMR model. 

First, we assume that 
\begin{equation}
  \lim_{n_i \to \infty} \frac{1}{n_i} \left(\Xv_{g,i}'\Xv_{g,i}-\Xv_{g,i}'\Xv_{c,i}(\Xv_{c,i}'\Xv_{c,i})^{-1}\Xv_{c,i}'\Xv_{g,i}\right) = \Qv_i,
\end{equation}
and $Q_i$ is also full-rank. Hence,
\begin{equation}
  \label{lim.vg}
   \lim_{n_i \to \infty}\Vv_g = \oplus_{i=1}^s \left[\frac{1}{n_i} \left(\Qv_i^{-1}\otimes \Sigv_i\right)\right].
\end{equation}  

When $\SSigv$ is known, as $n_i \to \infty $ for each $i$, based on (\ref{lim.vg})
\begin{equation}
  \label{bic.lim1}
  \lim_{n_i \to \infty, \forall i} \left(\Iv + \Vv_g^{-1} \Wv_g\right) =  \Vv_g^{-1} \Wv_g,
\end{equation}
and
\begin{equation}
   \label{bic.lim2}
   \lim_{n_i \to \infty, \forall i} \BF(\Wv_g) = |\Vv_g|^{1/2} \cdot |\Wv_g|^{-1/2} \cdot \exp\left(\frac{1}{2} \hbvg' \Vv_g^{-1} \hbvg \right).   
\end{equation}   
Note that
\begin{equation}
   \lim_{n_i \to \infty}|\Vv_g| = \prod_{i=1}^s \left( n_i^{-pr} \cdot |\Qv_i|^{-r} \cdot |\Sigv_i|^p \right), 
\end{equation}
and the likelihood ratio 
\begin{equation}
   L_1/L_0 = \frac{p(\YYv|\XXv,\hat \BBv, \SSigv)}{p(\YYv|\XXv,\tilde \BBv, \SSigv, H_0)} = \exp\left(\frac{1}{2} \hbvg' \Vv_g^{-1} \hbvg \right)
\end{equation}
It follows that
\begin{equation}
  \begin{aligned}
  \log \BF(\Wv_g) & \approx (\log L_1 - \log L_0) -\frac{pr}{2}\sum_{i=1}^s \log n_i  + \left(\frac{r}{2} \sum_{i=1}^s \log|\Qv_i| - \frac{p}{2}\sum_{i=1}^s \log|\Sigv_i| - \frac{1}{2} \log|\Wv_g| \right) \\
   & = (\log L_1 - \log L_0) -\frac{pr}{2}\sum_{i=1}^s \log n_i + O(1), \\
   & = {\rm BIC} + O(1).
  \end{aligned} 
\end{equation}

The BIC is asymptotically consistent, meaning that as sample size increases to infinity and under other suitable conditions, the BIC selects the fixed true model among a finite set of candidates with probability 1 (\cite{Haughton1988, Schwarz1978}). Consequently, our Bayes factor and its approximations also enjoy this asymptotic consistency property.  

It is worth pointing out that the BIC is not a universal approximation of Bayes factors. In our case, BIC fails to approximate desired Bayes factors with the advocated error bound if the pre-specified conditions are violated. In particular, 
\begin{enumerate}
 \item $\Wv_g$ or $\Vv_g$ is singular. Intuitively, in this case, linear constraints on parameter space would change the way that ``free" parameters are counted. Nonetheless, it is usually possible to resolve the linear constraints by transformation and re-parametrization.
 
 \item $\Wv_g$ is some function of sample sizes, e.g., this may lead that $\lim_{n_i \to 0} \frac{\log |\Wv_g|}{n_i} \ne 0, $ for some  $i$.  An example of this sort is the prior specification, $\Wv_g = c\Vv_g$. It is easy to see that BIC fails to approximate the resulting Bayes factor with the advocated error bound.
  
  \item Parameters $p,r$ and $s$ are {\em not} small comparing with sample sizes. In particular, under the high-dimensional settings, the BIC becomes a very poor approximation of the desired Bayes factor.
\end{enumerate}

When the $\SSigv$ is unknown, it can be shown that
\begin{equation}
  \log \ABF(\Wv_g, \av) \approx \frac{1}{2} \hbvg' \check \Vv_g^{-1} \hbvg  -\frac{pr}{2}\sum_{i=1}^s \log n_i  + \left(\frac{r}{2} \sum_{i=1}^s \log|\Qv_i| - \frac{p}{2}\sum_{i=1}^s \log|\check \Sigv_i| - \frac{1}{2} \log|\check \Wv_g| \right).
\end{equation}
In particular,
\begin{equation}
  \log \ABF(\Wv_g, \av=1) \approx \frac{1}{2} \hbvg' \hat \Vv_g^{-1} \hbvg  -\frac{pr}{2}\sum_{i=1}^s \log n_i  + \left(\frac{r}{2} \sum_{i=1}^s \log|\Qv_i| - \frac{p}{2}\sum_{i=1}^s \log|\hat \Sigv_i| - \frac{1}{2} \log|\hat \Wv_g| \right),
\end{equation}
Asymptotically, under the conditions stated
\begin{equation}
 \lim_{n_i \to \infty, \forall i} \, \hbvg' \hat \Vv_g^{-1} \hbvg \to \hbvg \Vv_g^{-1} \hbvg.
\end{equation}
Furthermore, it can be shown that
\begin{equation}
  \lim_{n_i \to \infty} \tilde \Sigv_i = \hat \Sigv_i + \hat \Bv_{g,i}'\Qv_i \hat \Bv_{g,i}.
\end{equation}
In general, this ensures that
\begin{equation}
  \hbvg' \check \Vv_g^{-1} \hbvg = \hbvg' \Vv_g^{-1} \hbvg + O(1).
\end{equation}
This yields our final results: under the conditions stated
\begin{equation}
   \log \ABF(\Wv_g, \av) = (\log L_1 - \log L_0) -\frac{pr}{2}\sum_{i=1}^s \log n_i + O(1).
\end{equation}

\section{Extension to Non-normal Data}

Without loss of generality, we consider a system of generalized linear models which resembles the SSLR. The MLE of the system can be numerically computed for the vectorized regression coefficients $\bvs$. Following the standard asymptotic maximum likelihood theory, the likelihood of the system can be approximated by a quadratic expansion around its maximum likelihood estimate. This can be equivalently expressed by the following asymptotic approximation,
\begin{equation}
   \label{glm.lik}
   \hbvs \,|\, \bvs \sim {\rm N}\left( \bvs\,,\, {\rm Var}(\hbvs) \right), 
\end{equation} 
where ${\rm Var}(\hbvs)$ is typically approximated using observed Fisher information. Combining with the prior distribution 
\begin{equation}
   \bvs \sim {\rm N}( {\bf 0}\,,\, \Psiv_c \oplus \Wv_g),
\end{equation}
it is then straightforward to show that the resulting Bayes factor under this setting maintains the same functional form as in LEMMA 1.

\section{MCMC Algorithm for Model Selection in MVLR} 

We implement an Markov Chain Monte Carlo (MCMC) algorithm to generate samples for posterior analysis of $\xi(\bvg)$. Here we detail the algorithm for the MVLR model, and point out that generalizing this algorithm for the general SSMR model is trivial.

\subsection{Description of Algorithm}

In the SSMR model, the posterior distribution of $\xi(\bvg)$ is given by
\begin{equation}
  \label{mcmc.post}
  \begin{aligned}
     \Pr(\xi(\bvg) \mid \Yv, \Xv) & \propto \Pr(\xi(\bvg)) \cdot p(\Yv \mid \xi(\bvg), \Xv) \\
                                            & \propto \Pr(\xi(\bvg)) \cdot \BF(\xi(\bvg)).
  \end{aligned}
\end{equation}
In the main text, we have discussed the computation of $\BF(\xi(\bvg))$. Assuming the prior distribution $\Pr(\xi(\bvg))$ is provided and easy to compute, it is straightforward to apply the Metropolis-Hastings algorithm. 
The practical difficulty in applying this algorithm in high-dimensional settings is to find an efficient proposal distribution to ensure the fast mixing of the Markov chain.  

In solving Bayesian variable selection problem in the multiple linear regression context, \cite{Guan2011} proposed a novel proposal distribution that prioritizes updates on variables showing strong marginal associations, an idea related to the {\em sure-independence screening} (\cite{Fan2008}). We generalize their idea in the context of the SSMR model. 
In our implementation, we utilize two types of simple ``local" proposal updates:  
\begin{enumerate}
\item changing the configuration of a candidate covariate.
\item swapping the configurations of two different covariates. 
\end{enumerate}
More specifically, each covariate $i$ is proposed according to a weight $w_i$ computed by
\begin{equation}
  w_i = \sum_{j=1}^{n-1} p_j {\rm BF}_i^{[j]} + p_n.
\end{equation}
The quantity ${\rm BF}_i^{[j]}$ represents the single-variate Bayes factor of covariate $i$ obtained by averaging (equally) over its all non-zero configuration Bayes factors and controlling for previously identified $(j-1)$ top association signals. We construct the weights by starting with an empty set of controlling covariates and compute the single covariate Bayes factors; we then select the covariate with the highest marginal Bayes factor into the set of covariates to be controlled for in the next round; we repeat this procedure $(n-1)$ times and in the $n$-th round, we simply assign each covariate uniform weight. Finally, we combine these weights into $w_i$ by a sequence of non-increasing probabilities $p_1 > p_2 > ... > p_n$. 
The general idea of this proposal distribution is largely due to Matthew Stephens (personal communication).
In the simulation and data application examples of this paper, we set $n=4$ and $p_1 = 0.624,~p_2 = 0.250,~p_3 = 0.125,~p_4 = 0.010$. In practice, once a SNP is proposed, we  randomly assign 85\% of the proposals to move type 1 and the 15\% of the proposals to move type 2.

In addition, when processing the posterior samples to compute posterior inclusion probabilities of covariates, we utilize Rao-Blackwellization techniques to reduce Monte Carlo variance of the estimates.

\subsection{Convergence Diagnostics}
	
We describe two convergence diagnostics of the proposed MCMC algorithm in this section. The first method is a direct adaption of \cite{Brooks2003}, which is a formal convergence testing procedure and requires running multiple chains. 
The other informal diagnostic we found useful is to utilize (\ref{mcmc.post}), which essentially is the posterior model probability up to a unknown normalizing constant. For each MCMC run, 
we compute the rank correlation between the posterior sampling frequencies and corresponding posterior scores for the sampled models. When the MCMC algorithm reaches convergence, we expect this correlation is high for the top ranked posterior models.
Our observation is that the rank correlation is indeed high, the formal testing of convergence usually becomes redundant and can be avoided. As a result, it reduces the computational burden to run multiple Markov chains.
  
\subsection{Computational Benchmark}

We benchmark the computational performance of the MCMC algorithm (implemented in C++) analyzing the imputed SNP data set of Gene C21orf57. The data set contains 4797 SNPs, 75 individuals and expression levels from three cell types. The program is running on a computer with 8-core Intel Xeon 2.13GHz processors and uses 25 Megabtypes of memory space. For 25000 burning steps and 50000 MCMC repeats, the full computation takes 14 minutes 22 seconds real time.

\section{Additional Simulation Results}

We perform additional simulation studies to fully investigate the difference in performance of BMS and LASSO. In the end, we identify two primary factors that may explain the observed performance patterns:
\begin{enumerate}
  \item the correlation structure of the random errors in the MVLR model.
  \item the prior correlation information of non-zero regression coefficients.
\end{enumerate}
Notably, vanilla version of the LASSO algorithm takes account of neither. To evaluate their individual effects on model selection, we simulate additional data for $n=100, p=250$ and $r=3$ under the MVLR model, for which all candidate covariates are independently generated.

\subsection{Impact of Error Variance Matrix}

We first investigate the impact of the error variance on model selection. To do so, we simulate independent regression coefficients across subgroups for each selected covariate, but alter the error variance matrix $\Sigma$ for the MVLR model. In particular, we use the following three different settings for the $\Sigma$ matrix:
\begin{enumerate}
  \item $\Sigv = \sigma^2 \Iv$
  \item $\Sigv$ is diagonal, but the diagonal elements are unequal (i.e., unequal error variances in different subgroups).
  \item $\Sigv$ has non-zero correlations between subgroups and unequal diagonal elements.
\end{enumerate}
In all three settings, we run both LASSO and BMS on the simulated data sets. For BMS, we assume the prior effect sizes are independent within each covariate in all cases; and in specifying the Wishart prior for $\Sigv$, we set $\Hv \to 0$ and $\nu \to 0$ (i.e. $\Sigv$ are directly estimated from the data with essentially no prior influence).

We plot the trade-off between the true positives and false positives from both methods in Figure \ref{supsim1.fig}. Our result indicates that when $\Sigv = \sigma^2 \Iv$, the two methods perform very similarly. However, as the true $\Sigma$ departs further away from the diagonal and equal variance structure, the performance of LASSO becomes worse. In comparison, the performance of BMS is stable in all three settings. 
\begin{figure}[h!t]
\begin{center}
\includegraphics[totalheight=0.30\textheight]{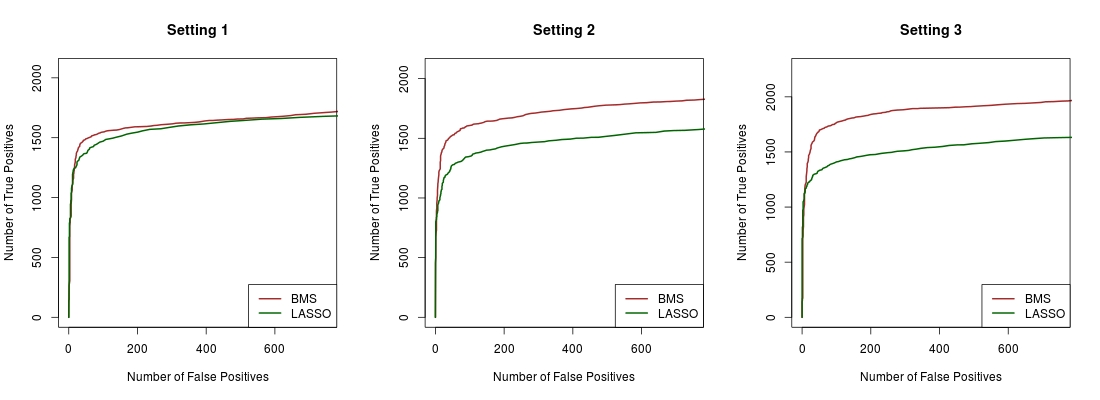}
\caption{\label{supsim1.fig}Evaluation of impact of error variance on model selection methods. In setting 1, the true $\Sigv = \sigma^2 \Iv$; in setting 2, the true $\Sigv$ is  diagonal but with unequal diagonal elements; in setting 3, the true $\Sigv$ has the most general form, with non-zero positive correlations and unequal diagonal elements. BMS has similar performance across the three settings, LASSO seems performing worse when the true $\Sigv$ departs further away from $ \sigma^2 \Iv$. }
 \end{center}
\end{figure}    

\cite{Rothman2010} also discovered the structure of $\Sigv$ matrix has significant impacts on the performance of regularized model selection method. As a remedy, they propose to regularize $\Sigv$ matrix jointly with $\bv$ in the $L_1$ penalty term. However in our context, $\Sigv$ is considered to be low dimensional ($r=3$) and the motivation to regularize $\Sigv$ is unclear to us.

\subsection{Importance of Prior Information}

We also examine the importance of utilizing prior correlation information of non-zero coefficients on the performance of model selection. Again, we limited our comparisons to BMS and LASSO using only simulated independent covariate data. Furthermore, we use $\Sigv = \sigma^2 \Iv$ to generate random errors for the MVLR model in this part of the simulation study. 

We create two different schemes in generating regression coefficients. The first scheme is the same as we described in the main text (i.e., conditioning on a non-zero configuration, $\gav = (111)$ is with probability 0.50 and others are equally likely). In the second scheme, we assign the activity configuration $\gav = (111)$ with probability 1 to the selected covariate. For the $i$th selected covariate, the effect sizes in the three subgroups are subsequently simulated from ${\rm N}(\bar \beta_i, \frac{{\bar \beta_i}^2}{100})$, where $\bar \beta_i$ is drawn from a ${\rm N}(0,1)$ distribution. The resulting correlation structure of regression coefficients is most similar to what have been observed in a meta-analysis.  

We run both BMS and LASSO on 200 data sets simulated in each scheme. To specify the distribution of non-zero activity configurations for BMS, we use both the default ``objective" prior (which assigns equal probability mass to each non-zero activity configuration) and the ``perfect" prior (which is the true generative distribution of the simulation data sets).

We show the simulation results in Figure \ref{supsim2.fig} by plotting the trade-off between the true positives and the false positives for each method. The results show that BMS with perfect prior information always achieves the best performance. (We again emphasize that in many genomic applications, it is possible to accurately estimate this ``perfect" prior from data, see examples from \cite{Flutre2012}). Although the ``objective" prior is clearly not optimal, because it captures the correlations between non-zero effects within a covariate, it still outperforms LASSO in both cases. Finally, we expect a prior assuming independence of effects of regression coefficients (i.e. a diagonal $\Gammav_g$ matrix) will behave similarly to LASSO, based on our observation in setting 1 of Figure \ref{supsim1.fig}. 
Therefore, we conclude that the performance of model selection methods are likely to have significant improvement if the {\it a priori} information in data can be accurately utilized.   

\begin{figure}[h!t]
\begin{center}
\includegraphics[totalheight=0.32\textheight]{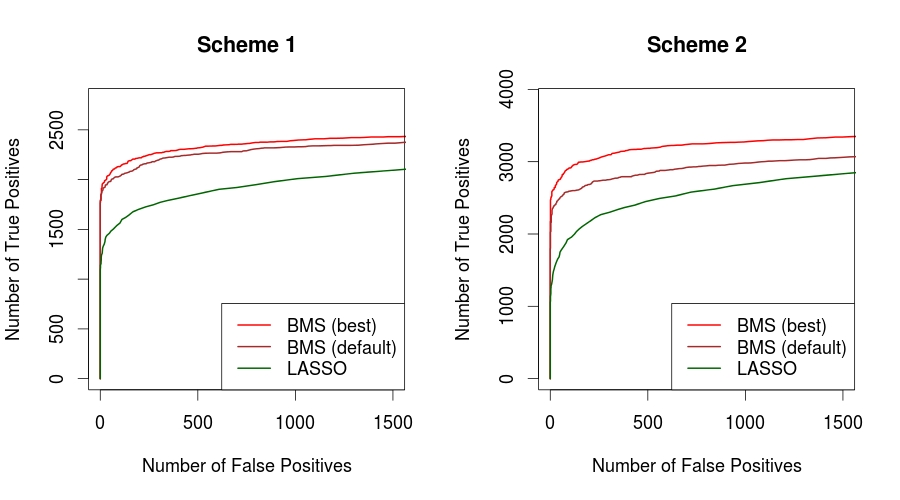}
\caption{\label{supsim2.fig}Evaluation of impact of prior information on model selection methods. Scheme 1 and 2 correspond to two distinct generating distributions used for simulating data. BMS(best) is our Bayesian model selection method using the true generative distribution as the prior, whereas BMS(default) uses an ``objective" prior. In scheme 1, the objective prior is ``closer" to the truth than in scheme 2. LASSO does not utilize the prior correlation information and essentially assumes that the regression coefficients are {\it a priori} independent.}
 \end{center}
\end{figure}

\section{Single SNP Analysis Result for Gene C21orf57}  

In this section, we show the single SNP analysis results of the eQTL mapping for gene C21orf57 using Dimas data. More specifically, the aim is to examine the results of the tissue specificity inference from our BMS approach. 

As a visual diagnostic, we first fit a simple linear regression model for each SNP in each cell type, we then examine the resulting regression coefficients across all three cell types for each SNP using a forest plot. We show the results for the three distinct signals identified by the BMS approach in Figure \ref{forest.figure}. By this simple diagnostic, the tissue specificity inference seems intuitively sensible. 

\begin{figure}[h!t]
\begin{center}
\includegraphics[totalheight=0.4\textheight]{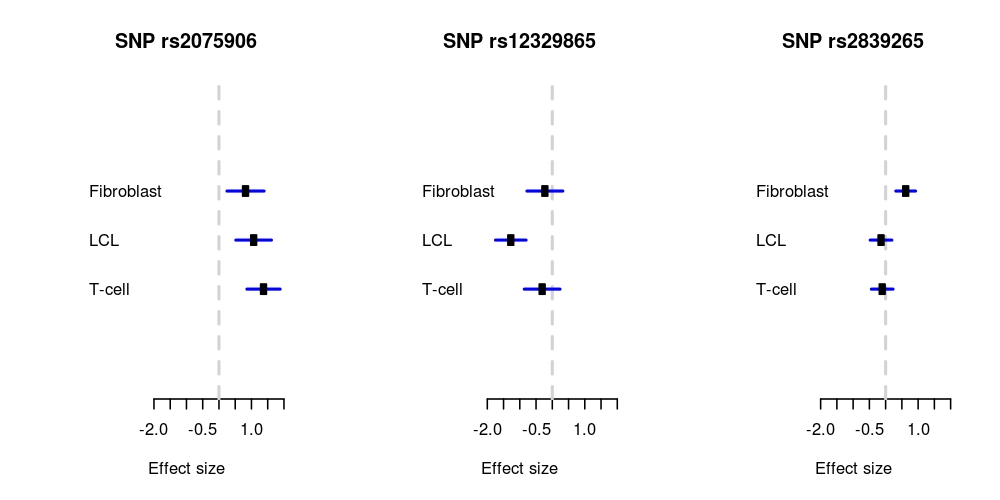}
\caption{\label{forest.figure} Examining the single-SNP effects of each identified eQTL in each separate cell type for Gene C21orf57. For each SNP in each cell type, we obtain the estimates of the effect size and its standard error by fitting a simple linear regression model. The estimated effect sizes and their corresponding 95\% confidence intervals are plotted for different cell types in a forest plot for each SNP. SNP rs2839265 has shorter intervals because its minor allele frequency (0.28) is greater than the frequencies of rs2075906 (0.09) and rs12329865 (0.11). Recall our method infers that rs2075906 is tissue consistent; SNPs rs12329865 and rs2839265 are LCL-specific and Fibroblast-specific eQTLs, respectively.}
\end{center}
\end{figure}	
 
\bibliographystyle{natbib}  
\bibliography{bvscls_arxiv}

\end{document}